\documentclass[12pt]{article}
\usepackage[utf8]{inputenc}

\usepackage{amssymb,amsmath,amsthm}
\usepackage{url}
\usepackage{xspace}
\usepackage{enumerate}

\usepackage[dvipsnames,usenames]{xcolor}
\usepackage{tikz}
\usetikzlibrary{fit}
\usetikzlibrary{shapes}
\usetikzlibrary{backgrounds}
\usetikzlibrary{patterns}
\usetikzlibrary{calc}
\usetikzlibrary{graphs}

\usepackage{hyperref}

\theoremstyle{plain}
\newtheorem{theorem}{Theorem}[section]
\newtheorem{lemma}[theorem]{Lemma}
\newtheorem{corollary}[theorem]{Corollary}
\theoremstyle{definition} % This does not use italics for the text.

\newtheorem{definition}[theorem]{Definition}

\newcommand{\Case}[2]{\smallskip\par{\it Case #1:\/ #2}}

\newcommand{\refeq}[1]{(\ref{eq:#1})}

\newcommand{\of}[1]{\left( #1 \right)}
\newcommand{\Set}[1]{\left\{\, #1 \,\right\}}
\newcommand{\setdef}[2]{\left\{ \hspace{0.5mm} #1 : \hspace{0.5mm} #2 \right\}}
\newcommand{\msetdef}[2]{\left\{\!\!\left\{ \hspace{0.5mm} #1 : \hspace{0.5mm} #2 \right\}\!\!\right\}}

\newcommand{\A}{\forall}
\newcommand{\E}{\exists}

\newcommand{\und}{\wedge}

\newcommand{\function}[2]{:#1 \rightarrow #2}

\newcommand{\WL}[1]{\ensuremath{#1\text{-}\mathrm{WL}}\xspace}
\newcommand{\kWL}{\WL k}
\newcommand{\sub}[2]{\mathrm{sub}(#1,#2)}

\newcommand{\hhomo}[2]{\mathrm{hom}(#1,#2)}
\newcommand{\eqkwl}{\equiv_{\kWL}}
\newcommand{\eqkkwl}[1]{\equiv_{#1\text{-}\mathrm{WL}}}

\newcommand{\cclass}[1]{\ensuremath{\mathcal{C}(#1)}\xspace}
\newcommand{\rclass}[1]{\ensuremath{\mathcal{R}(#1)}\xspace}

\newcommand{\forb}[1]{\mathit{Forb}(#1)}

\newcommand{\Pa}{{\mathcal P}}

\newcommand{\colclass}[1]{
\makeatletter
\xdef\clcl@verts{}
\foreach \v in {#1}{\xdef\clcl@verts{\clcl@verts(\v)}}
\begin{scope}[on background layer]
\node[fit/.expand once={\clcl@verts},inner sep=4pt,rectangle,fill=black!20,rounded corners=8pt,draw=none] {};
\end{scope}
\makeatother
}

\newcommand{\mset}[1]{\left\{\!\!\left\{ \hspace{0.5mm} #1 \hspace{0.5mm} \right\}\!\!\right\}}
\newcommand{\alg}[4]{\mathrm{WL}_{#1}^{#2}(#3,#4)}
\newcommand{\algg}[3]{\mathrm{WL}_{#1}^{#2}(#3)}
\newcommand{\baru}{{\bar u}}
\newcommand{\tw}[1]{\mathit{tw}(#1)}
\newcommand{\htw}[1]{\mathit{htw}(#1)}
\newcommand{\probl}[1]{\textsc{\small #1}}

\title{On Weisfeiler-Leman Invariance:\\
Subgraph Counts and Related Graph Properties}

\author{V.~Arvind\thanks{The Institute of Mathematical Sciences (HBNI), Chennai, India.}, 
Frank Fuhlbrück\thanks{Institut für Informatik, Humboldt-Universität zu Berlin, Germany.}, 
Johannes Köbler${}^\dagger$, Oleg Verbitsky${}^\dagger$\,\thanks{Supported by DFG grant KO 1053/8--1. On leave from the IAPMM, Lviv, Ukraine.}}

\date{}

\begin{document}

\maketitle

\begin{abstract}
The $k$-dimensional Weisfeiler-Leman algorithm (\kWL) is a fruitful
approach to the Graph Isomorphism problem. \WL2 corresponds to the
original algorithm suggested by Weisfeiler and Leman over 50 years
ago. \WL1 is the classical color refinement routine.
Indistinguishability by \kWL is an equivalence relation on graphs that
is of fundamental importance for isomorphism testing, descriptive
complexity theory, and graph similarity testing which is also of some
relevance in artificial intelligence.  Focusing on dimensions
$k=1,2$, we investigate subgraph patterns whose counts are \kWL
invariant, and whose occurrence is \kWL invariant. We achieve a
complete description of all such patterns for dimension $k=1$ and
considerably extend the previous results known for $k=2$.
\end{abstract}

\section{Introduction}

\emph{Color refinement} is a classical procedure widely used in
isomorphism testing and other areas.  It initially colors each vertex
of an input graph by its degree and refines the vertex coloring in
rounds, taking into account the colors appearing in the neighborhood
of each vertex. This simple and efficient procedure successfully
canonizes almost all graphs in linear time \cite{BabaiES80}. Combined
with \emph{individualization}, it is the basis of the most successful
practical algorithms for the graph isomorphism problem; see
\cite{McKayP14} for an overview and historical comments.

The first published work on color refinement dates back at least to
1965 (Morgan \cite{Morgan65}). In 1968 Weisfeiler and Leman
\cite{WLe68} gave a procedure that assigns colors to \emph{pairs} of
vertices of the input graph. The initial colors are \emph{edge},
\emph{nonedge}, and \emph{loop}. The procedure refines the coloring in
rounds by assigning a new color to each pair $(u,v)$ depending on the
color types of the 2-walks $uwv$, where $w$ ranges over the vertex
set. The procedure terminates when the color partition of the set of
all vertex pairs stabilizes. The output coloring is an isomorphism
invariant of the input graph. It yields an edge-colored complete
directed graph with certain highly regular properties. This object,
known as a \emph{coherent configuration}, has independently been
discovered in other contexts in statistics (Bose \cite{BoseM59}) and
algebra (Higman~\cite{Higman64}).

A natural extension of this idea, due to Babai (see
\cite{Babai16,CaiFI92}), is to iteratively classify $k$-tuples of
vertices. This is the \emph{$k$-dimensional Weisfeiler-Leman
  procedure}, abbreviated as \kWL. Thus, \WL2 is the original
Weisfeiler-Leman algorithm \cite{WLe68}, and \WL1 is color refinement.
The running time of \kWL is $n^{O(k)}$, where $n$ denotes the number
of vertices in an input graph. Cai, Fürer, and Immerman \cite{CaiFI92}
showed that there are infinitely many pairs of nonisomorphic graphs
$(G_i,H_i)$ such that $\kWL$ fails to distinguish between them for any
$k=o(n)$. Nevertheless, the Weisfeiler-Leman procedure, as an
essential component in isomorphism testing, can hardly be
overestimated. A constant dimension often suffices to solve the
isomorphism problem for important graph classes. A striking result
here (Grohe \cite{Grohe12}) is that for any graph class excluding a
fixed minor (like bounded genus or bounded treewidth graphs)
isomorphism can be tested using $\kWL$ for a constant $k$ that only
depends on the excluded minor. Moreover, Babai's quasipolynomial-time
algorithm \cite{Babai16} for general graph isomorphism crucially uses
$\kWL$ for logarithmic~$k$.

We call two graphs $G$ and $H$ \emph{\kWL-equivalent} and write
$G\eqkwl H$ if they are indistinguishable by \kWL; formal definitions
are given in Sections \ref{s:CR1} ($k=1$) and \ref{s:Ck} ($k\ge2$).
By the Cai-Fürer-Immerman result~\cite{CaiFI92}, we know for any given
$k$ that the $\eqkwl$-equivalence is coarser than the isomorphism
relation on graphs.

\begin{definition}
A graph property (i.e., an isomorphism-invariant family of graphs)
$\mathcal{P}$ is \emph{\kWL-invariant} if for any pair of graphs $G$
and $H$:
$$
G\in\mathcal{P}\text{ and }G\eqkwl H\text{ implies }H\in\mathcal{P}.
$$
In particular, a graph parameter
$\pi$ is \emph{\kWL-invariant} if 
$\pi(G)=\pi(H)$ whenever $G\eqkwl H$.
\end{definition}

The broad question of interest in this paper is which graph properties
(and graph parameters) are \kWL-invariant for a specified~$k$.
The motivation for this natural question comes from various areas.
Understanding the power of \kWL, even for small values of $k$, is
important for both isomorphism testing and graph similarity testing.
For example, the largest eigenvalues of \WL1-equivalent graphs are
equal~\cite{ScheinermanU97}. Moreover, \WL2-equivalent graphs are
cospectral \cite{DawarSZ17,Fuerer10}. Consequently, by Kirchhoff's
theorem, \WL2-equivalent graphs have the same number of spanning
trees. Also the \WL2-invariance of certain metric graph parameters
such as diameter is easy to show.  Fürer \cite{Fuerer17} recently
asked which basic combinatorial parameters are \WL2-invariant. While
it is readily seen that \WL2-equivalence preserves the number of
3-cycles, F\"urer pointed out, among other interesting observations,
that also the number of $s$-cycles is \WL2-invariant for
each $s\le6$. More recently, Dell, Grohe, and Rattan \cite{DellGR18}
characterized \kWL-equivalence in terms of homomorphism profiles.
Specifically, they show that $G\eqkwl H$ if and only if the number of
homomorphisms from $F$ to $G$ and to $H$ are equal for all graphs $F$
of treewidth at most~$k$.

As a heuristic for graph similarity testing, the Weisfeiler-Leman
procedure has been applied in artificial intelligence; see
\cite{ShervashidzeSLMB11} for an \WL1-based application and
\cite{MorrisKM17} for a multidimensional version. It is noteworthy
that \WL1 turns out to be exactly as powerful as graph neural networks
\cite{Morris+18}.  Comparing subgraph frequencies is also widely used
for testing graph similarity and detecting structure of large
real-life graphs; see,
e.g.,~\cite{GaoL17,GrochowK07,Milo+02,UganderBK13}. For example, just
knowing the number of triangles is valuable information about a social
network; see, e.g.,~\cite{HasanD18}. Important structural information
can also be found from the number of paths of length 2 and from the
degree distribution, i.e., the statistics of star subgraphs;
see~\cite{Newman03}. This poses a natural question on how much the
two approaches --- one based on \kWL-equivalence and one based on
subgraph statistics --- are related to each other.

Finally, \kWL-equivalence 
is of fundamental importance for finite and algorithmic
model theory.  
A graph property $\mathcal{P}$ is
\kWL-invariant exactly when $\mathcal{P}$ is definable in the
$(k+1)$-variable infinitary counting logic.
Showing a graph property $\mathcal{P}$ to be not $\eqkwl$-invariant
for any $k$ will imply $\mathcal{P}$ is not definable in fixed-point
logic with counting (FPC); see, e.g., the survey \cite{Dawar15}.  A
systematic study of \kWL-invariant constraint satisfaction problems
was undertaken by Atserias, Bulatov, and Dawar~\cite{AtseriasBD09}.

\subsection*{Our results}

Let $F$ be a fixed pattern graph and $G$ be any given graph. The main
focus of our paper is to investigate the \kWL-invariance of: (a) the
property that $G$ contains $F$ as a subgraph, and (b) the number of
subgraphs of $G$ isomorphic to $F$. We use \emph{$\sub F\cdot$} to
denote the subgraph count function.  Thus, $\sub FG$ denotes the
number of subgraphs of $G$ isomorphic to~$F$.

\begin{definition}\label{def:crclasses}
   Let \cclass k denote the class of all pattern graphs $F$ for which
   the subgraph count $\sub F\cdot$ is $\eqkwl$-invariant.
   Furthermore, \rclass k consists of all pattern graphs $F$ such that
   the property of a graph containing $F$ as a subgraph is
   $\eqkwl$-invariant.
\end{definition}

The concepts of \cclass k and \rclass k correspond
to algorithmic \emph{counting} and \emph{recognition} problems respectively.
Note that $\cclass k\subseteq\rclass k$. We use this notation to state
some consequences of prior work. The \kWL-equivalence
characterization~\cite{DellGR18}, stated above, can be used to show
that \cclass k contains every $F$ such that all homomorphic images of
$F$ have treewidth no more than $k$. We say that such an $F$ has
\emph{homomorphism-hereditary treewidth} at most $k$; see Section
\ref{s:Ck} for details.  The striking result by Anderson, Dawar, and
Holm \cite{AndersonDH15} on the expressibility of the matching number
in FPC
%, along with the well-known fact that each formula of this logic
% can be rewritten in finite-variable infinitary counting logic, 
implies
that there is some $k$ such that \rclass k contains \emph{all}
matching graphs $sK_2$, where $sK_2$ denotes the disjoint union of $s$
edges.  On the other hand, there is no $k$ such that \rclass k
contains all cycle graphs $C_s$. This readily follows from the result
by Dawar~\cite{Dawar98,Dawar15} that the property of a graph having a
Hamiltonian cycle is not $\eqkwl$-invariant for any~$k$.

Our results are as follows.

\paragraph{\textit{Complete description of \cclass1 and \rclass1
(invariance under color refinement).}}

We prove that, up to adding isolated vertices, \cclass1 consists of
all star graphs $K_{1,s}$ and the 2-matching graph~$2K_2$. Hence,
\cclass1 contains exactly the pattern graphs of homomorphism-hereditary
treewidth equal to 1. Another noteworthy consequence is that, for
every $F\in\cclass1$, the subgraph count $\sub FG$ is determined just
by the degree sequence of a graph~$G$.

We obtain a complete description of \rclass1 by proving that this
class consists of the graphs in \cclass1 and three forests $P_3+P_2$,
$P_3+2P_2$, and $2P_3$, where $P_s$ denotes the path graph on $s$
vertices.

\paragraph{\textit{Case study for \cclass2 and \rclass2
(invariance under the original Weisfeiler-Leman algorithm).}}

An explicit characterization of \cclass2 and \rclass2 appears
challenging.  Indeed, it is not a priori clear whether testing
membership in these graph classes is possible in polynomial time.
While it is unknown whether \cclass2 consists exactly of graphs with
homomorphism-hereditary treewidth bounded by $2$, we prove that this
is indeed the case for some standard graph sequences.  These results
are related to questions that have been discussed in the literature.

\begin{itemize}
\item 
Beezer and Farrell \cite{BeezerF00} proved that the first five
coefficients of the matching polynomial of a strongly regular graph
are determined by its parameters.\footnote{The result of \cite{BeezerF00}
is actually stronger and applies even to distance-regular graphs: The first five
coefficients of the matching polynomial of such a graph
are determined by the intersection array of the graph.}
I.e., if $G$ and $H$ are strongly regular graphs with the same
parameters, then $\sub{sK_2}G=\sub{sK_2}H$ for $s\le5$.  We prove that
$sK_2\in\cclass2$ if and only if $s\le5$. It follows that the
Beezer-Farrell result extends to $\WL2$-equivalent graphs. I.e., if
$G$ and $H$ are any two $\WL2$-equivalent graphs, then the first five
coefficients of their matching polynomials coincide. Moreover, this
result is tight and cannot be extended to a larger~$s$.  Note that
strongly regular graphs with the same parameters are the simplest
example of \WL2-equivalent graphs.\footnote{Two distance-regular graphs
with the same intersection array are also \WL2-equivalent.}
\item 
Fürer~\cite{Fuerer17} proved that $C_s\in\cclass2$ for $3\le s\leq 6$
and $C_s\notin\cclass2$ for $8\le s\leq 16$.  We close the gap and
show that $C_7$ is the largest cycle graph in \cclass2.  We also prove
that \cclass2 contains $P_1,\ldots,P_7$ and no other path graphs.
The result on cycles admits the following generalization. First, we
observe that the girth $g(G)$ of a graph $G$ is a \WL2-invariant
parameter.  Then, we prove that if $G\eqkkwl2H$, then
$\sub{C_s}G=\sub{C_s}H$ for each $3\le s\le 2\,g(G)+1$. Neither the
factor of 2, nor the additive term of 1 can here be improved.
\end{itemize}

\noindent
Characterization of \rclass2 appears to be still harder. Fürer
\cite{Fuerer17} has shown that \rclass2 does not contain the complete
graph with 4 vertices. Building on that, we show that \rclass2 also
does not contain any graph $F$ with a unique 4-clique.  In view of
this result, it is natural to conjecture that \rclass2 does not
contain any graph of clique number more than~3.
We also show that
\rclass2 contains only finitely many cycle graphs $C_s$. Moreover,
following Dawar's approach~\cite{Dawar98}, for each $k$ we show that
\rclass k contains only finitely many~$C_s$.

\paragraph{Notation.}
% moved to the first appearence (in Appendix)
% The \emph{vertex cover number} $\tau(G)$ is
% equal to the minimum size of a vertex cover in $G$.  
The \emph{girth} $g(G)$ is the minimum length of a cycle in $G$.
If $G$ is acyclic, then $g(G)=\infty$. We denote the
vertex set of $G$ by $V(G)$ and the edge set by $E(G)$.  Furthermore,
$v(G)=|V(G)|$ and $e(G)=|E(G)|$. The set of vertices adjacent to a
vertex $u\in V(G)$ forms its neighborhood $N(u)$.  
The subgraph of $G$
induced by a subset of vertices $X\subseteq V(G)$ is denoted by
$G[X]$. For two disjoint vertex subsets $X$ and $Y$, we denote by
$G[X,Y]$ the bipartite graph with vertex classes $X$ and $Y$ and all
edges of $G$ with one vertex in $X$ and the other in $Y$.  
The vertex-disjoint union of graphs $G$ and $H$ is denoted by
$G+H$. Furthermore, we write $mG$ for the disjoint union of $m$ copies
of $G$.  We use the standard notation $K_n$ for complete graphs, $P_n$
for paths, and $C_n$ for cycles on $n$ vertices.  Furthermore,
$K_{s,t}$ denotes the complete bipartite graph whose vertex classes
have $s$ and $t$ vertices. Likewise, $K_{1,1}=K_2=P_2$,
$K_{1,2}=P_3$, $C_3=K_3$ etc.

\section{Color refinement invariance}\label{s:CR1}

Given a graph $G$, the \emph{color-refinement} algorithm (abbreviated
as \WL1) iteratively computes a sequence of colorings $C^i$ of $V(G)$.
The initial coloring $C^0$ is monochromatic, that is $C^0(u)$ is the
same for all vertices~$u$.  Then,
\begin{equation}
  \label{eq:Ci}
 C^{i+1}(u)=\left(C^i(u),\msetdef{C^i(a)}{a\in N(u)}\right),
\end{equation}
where $\{\!\!\{ \ldots \}\!\!\}$ denotes a multiset (i.e., the multiplicity of each element counts). 

If $\phi$ is an isomorphism from $G$ to $H$, then a straightforward
inductive argument shows that $C^i(u)=C^i(\phi(u))$ for each vertex
$u$ of~$G$.  This readily implies that, if graphs $G$ and $H$ are
isomorphic, then
\begin{equation}
  \label{eq:CR-check}
 \msetdef{C^i(u)}{u\in V(G)} = \msetdef{C^i(v)}{v\in V(H)}
\end{equation}
for all $i\ge0$.  
We write $G\eqkkwl1H$ exactly when this condition is met.

The following fact is a direct consequence of the definition.

\begin{lemma}\label{lem:AABB}
If $A\eqkkwl1B$ and $A'\eqkkwl1B'$, then $A+A'\eqkkwl1B+B'$.
\end{lemma}

\WL1 \emph{distinguishes} graphs $G$ and $H$ if $G\not\eqkkwl1H$.
In fact, the algorithm does not need to check \refeq{CR-check} for
infinitely many $i$: If Equality~\refeq{CR-check} is false for some
$i$ then it is false for $i=n$, where $n$ denotes the number of
vertices in each of the graphs.  
By this reason, we call the coloring $C^n$ \emph{stabilized}.

The partition $\Pa_G$ of $V(G)$ into color classes of $C^n$ is called
the \emph{stable partition} of $G$. We call the elements of $\Pa_G$
\emph{cells}.

The stable partition $\Pa_G$ is \emph{equitable}. I.e., for any two
(possibly equal) cells $X$ and $Y$, all vertices in $X$ have equally
many neighbors in $Y$ and vice versa. The number of neighbors that a
vertex of $X$ has in $Y$ will be denoted by $d(X,Y)$.  Thus, for each
cell $X$ the graph $G[X]$ induced by $X$ is \emph{regular}, that is,
all vertices in $G[X]$ have the same degree, namely $d(X,X)$.
Moreover, for all pairs of cells $X,Y$ the bipartite graph $G[X,Y]$
induced by $X$ and $Y$ is \emph{biregular}, that is, all vertices in
$X$ have equally many neighbors in $Y$ and vice versa.

The \emph{degree matrix} of $\Pa_G$ is defined as
$$
D_G=\big(d(X,Y)\big)_{X,Y\in\Pa_G}
$$ 
and indexed by the stabilized colors of the cells; that is, the index
$X$ of $D_G$ is the color $C^n(x)$ of any vertex $x\in X$.

\begin{lemma}\label{lem:degree}
$G\eqkkwl1H$ if and only if $D_G=D_H$.  
\end{lemma}

Indeed, the equality $D_G=D_H$ readily implies the equality \refeq{CR-check} for
$i=n$.  On the other hand, the inequality $D_G\ne D_H$ implies that
the multisets of colors in \refeq{CR-check} for $i=n$ are different. If they
were the same, then they would become distinct in the next refinement
round, i.e., for $i=n+1$ (whereas we know that if $\WL1$ can detect
such a distinction, it is detected by the $n$-th round).

Let $F$ be a graph and $s$ be a positive integer.  Note that $F$
belongs to \cclass k or \rclass k if and only if the graph $F+sK_1$
belongs to this class. Therefore, we will ignore isolated vertices.

\begin{theorem}\label{thm:CR1}
Up to adding isolated vertices, the classes \cclass1 and \rclass1
are formed by the following graphs.

\begin{enumerate}[\bf 1.]
\item  
\cclass1 consists of the
star graphs $K_{1,s}$ for all $s\ge1$ and the 2-matching graph~$2K_2$.
\item 
\rclass1 consists of the graphs in \cclass1 and the following three
forests:
\begin{equation}
  \label{eq:forests}
P_3+P_2,\,P_3+2P_2,\text{ and }2P_3.  
\end{equation}
\end{enumerate}
\end{theorem}

\noindent
The proof is spread over the next four subsections.

\subsection{Membership in \cclass1}

If two graphs are indistinguishable by color refinement, they have the
same degree sequence.  Notice that
$$
\sub{K_{1,s}}G=\sum_{v\in V(G)}{\deg v\choose s},
$$
where $\deg v$ denotes the degree of a vertex $v$.
This equality shows that $K_{1,s}\in\cclass1$.
Since any two edges constitute either $2K_2$ or $K_{1,2}$, we have
\begin{equation}\label{eq:2K2}
\sub{2K_2}G={e(G)\choose2}-\sub{K_{1,2}}G.  
\end{equation}
Taking into account that $e(G)=\frac12\sum_{v\in V(G)}\deg v$,
this implies that $2K_2\in\cclass1$.

Note that the equality \refeq{2K2} has been reported in several sources;
see, e.g., \cite[Lemma~1]{FarrellGC91} and the comments therein.

\subsection{Non-membership in \cclass1}

To prove that a graph $F$ is not in \cclass1, one needs to exhibit
\WL1-equivalent graphs $G$ and $H$ such that $\sub FG\ne\sub FH$. 
% For each of the 3 forests in~\refeq{forests} we can easily find
% witnesses $G$ and $H$ that are regular graphs with the same number of
% vertices and of the same degree. Specifically,
% $\sub{P_3+P_2}{C_6}=12$, while $\sub{P_3+P_2}{2C_3}=18$;
% $\sub{2P_3}{C_6}=3$, while $\sub{2P_3}{2C_3}=9$; and
% $\sub{P_3+2P_2}{C_7}=7$, while $\sub{P_3+2P_2}{C_4+C_3}\allowbreak=6$.  
Table \ref{fig:not-in-C1} provides a list of such
witnesses for each of the three forests in~\refeq{forests}.
The non-membership of all other graphs in \cclass1 follows from their
non-membership in \rclass1, which will be proved in the corresponding
subsection below.

\begin{table}[h]
\begin{center}
\renewcommand{\arraystretch}{1.0}
\begin{tabular}{l|lllll}
$F$       & $P_3+P_2$ & $2P_3$ & $P_3+2P_2$ \\
\hline
$G$       & $C_6$        & $C_6$     & $C_7$         \\
$\sub FG$ & 12           & 3         & 7            \\
$H$       & $2C_3$       & $2C_3$    & $C_4+C_3$     \\
$\sub FH$ & 18           & 9         & 6         
\end{tabular}
\end{center}
\caption{Witnesses to non-membership in \cclass1: Each pair $G$ and $H$ consists
of regular graphs with the same number of vertices and of the same degree.}\label{fig:not-in-C1}
\end{table}

\subsection{Membership in \rclass1}

We call a graph $H$ \emph{amenable} if color refinement distinguishes
$H$ from any other nonisomorphic graph $G$. For each of the three
forests $F$ in~\refeq{forests}, we are able to explicitly describe the
class $\forb F$ of $F$-free graphs. Based on this description, we can show that,
with just a few exceptions, every $F$-free graph is amenable.

\begin{lemma}\label{lem:F-free-am}\hfill
  \begin{enumerate}[\bf 1.]
  \item 
    Every $(P_3+P_2)$- or $2P_3$-free graph $H$ is amenable.
  \item 
  Every $(P_3+2P_2)$-free graph $H$ is amenable unless $H=2C_3$ or $H=C_6$.
  \end{enumerate}
\end{lemma}

Proving that $F\in\rclass1$ means proving the following implication:
\begin{equation}\label{eq:forb}
G\eqkkwl1H\ \ \&\ \ H\in\forb F \implies G\in\forb F.
\end{equation}
This implication is trivial whenever $H$ is an amenable graph because
then $G\cong H$. By Part 1 of Lemma \ref{lem:F-free-am}, we
immediately conclude that the graphs $P_3+P_2$ and $2P_3$ are in
$\rclass1$.  Part 2 ensures \refeq{forb} for all $(P_3+2P_2)$-free
graphs except $2C_3$ and $C_6$. However, the implication \refeq{forb}
holds true also for each exceptional graph $H\in\{2C_3,C_6\}$ by the
following trivial reason.  Since $H$ has 6 vertices, any
\WL1-indistinguishable graph $G$ must have also 6 vertices and hence
cannot contain a $P_3+2P_2$ subgraph.

The proof of Lemma \ref{lem:F-free-am} is lengthy and relies on
an explicit description of the class of $F$-free graphs for each
$F\in\{P_3+P_2,2P_3,P_3+2P_2\}$. Obtaining such a description requires
a scrupulous combinatorial analysis, and we postpone the proof to Appendix~\ref{app}.

\subsection{Non-membership in \rclass1}

We begin with proving that \rclass1 can contain only forests of stars.

\begin{lemma}%
[{see Bollob\'as \cite[Corollary 2.19]{Bollobas-b} or  Wormald \cite[Theorem~2.5]{Wor99}}]%
\label{lem:random-reg}
Let $d,g\ge3$ be fixed, and $dn$ be even.
Let $\mathcal{G}_{n,d}$ denote a random $d$-regular graph on $n$ vertices.
Then the probability that
$\mathcal{G}_{n,d}$ has girth $g$ converges to a non-zero limit as $n$ grows large.
\end{lemma}

\begin{lemma}\label{lem:no-cycle}
\rclass1 can contain only acyclic graphs.
\end{lemma}

\begin{proof}
  Assume that a graph $F$ has a cycle of length $m$. We show that it
  cannot belong to \rclass1.  Let $d=v(F)-1$.  Lemma
  \ref{lem:random-reg} ensures that there exists a $d$-regular graph
  $X$ of girth strictly more than $m$. Then $F$ does not appear as a
  subgraph in $H=(d+1)X$ but clearly does in $G=v(X)\,K_{d+1}$. It
  remains to notice that $G$ and $H$ are both $d$-regular and have the
  same number of vertices.
\end{proof}

\begin{lemma}\label{lem:only-sf}
\rclass1 can contain only forests of stars.  
\end{lemma}

\begin{proof}
Suppose that $F\in\rclass1$. By Lemma \ref{lem:no-cycle},
$F$ is a forest. In order to prove that every connected
component of $F$ is a star, it is sufficient and necessary
to prove that $F$ does not contain $P_4$ as a subgraph.
Assume, to the contrary, that $F$ has $P_4$-subgraphs.

Let $T$ be a connected component of $F$ containing $P_4$.  Consider a
diametral path $v_1v_2v_3\ldots v_d$ in $T$, where $d\ge4$. Note that
$v_1$ is a leaf. Let $T'$ be obtained from $T$ by identifying the
vertices $v_1$ and $v_4$. Thus, $T'$ is a unicyclic graph, where the
vertices $v_2$, $v_3$, and $v_4=v_1$ form a cycle $C_3$.  Obviously,
$v(T')<v(T)$.

Consider now the graph $H_T=2T'$. Identify one component of $H_T$ with
$T'$ and fix an isomorphism $\alpha$ from this to the other component
of $H_T$.  Let $G_T$ be obtained from $H_T$ by removing the edges
$v_2v_4$ and $\alpha(v_2)\alpha(v_4)$ and adding instead the new edges
$v_2\alpha(v_4)$ and $v_4\alpha(v_2)$. Note that, by construction,
$V(T')\subset V(H_T)=V(G_T)$.
Note that $G_T$ contains a subgraph isomorphic to~$T$.
We now prove that
\begin{equation}\label{eq:GTHT}
G_T\eqkkwl1H_T. 
\end{equation}

Indeed, define a map $\phi\function{V(G_T)}{V(T')}$ by
$\phi(u)=\phi(\alpha(u))=x$ for each $u\in V(T')\subset V(G_T)$.  Note
that $\phi$ is a \emph{covering map} from $G_T$ to $T'$, that is, a
surjective homomorphism whose restriction to the neighborhood of each
vertex of $G_T$ is surjective.  A straightforward inductive argument
shows that $\phi$ preserves the coloring produced by \WL1, that is,
$C^i(\phi(u))=C^i(u)$ for all $i$, where $C^i$ is defined by
\refeq{Ci}.  Thus, the multiset $\msetdef{C^i(u)}{u\in V(G_T)}$ is
obtained from the multiset $\msetdef{C^i(u)}{u\in V(T')}$ by doubling
the multiplicity of each color. Since $H_T$ consists of two disjoint
copies of $T'$, this readily implies that $G_T$ and $H_T$ are
indistinguishable by \WL1, and \refeq{GTHT} follows.

If a connected component $T$ of $F$ does not contain $P_4$, we set
$G_T=H_T=2T$. The equivalence \refeq{GTHT} is true also in
this case.  Define $G=\sum_TG_T$ and $H=\sum_TH_T$ where the disjoint
union is taken over all connected components $T$ of $F$.  We have
$G\eqkkwl1H$ by Lemma \ref{lem:AABB}.  Since each $G_T$ contains a
subgraph isomorphic to $T$, the graph $G$ contains a subgraph
isomorphic to $F$.  On the other hand, $H$ does not contain any
subgraph isomorphic to~$F$.  To see this, let $F_0$ be a non-star
component of $F$ with maximum number of vertices. Then $H$ cannot
contain even $F_0$ because every non-star component of $H$ has fewer
vertices than $F_0$.  Thus, we get a contradiction to the
assumption that $F\in\rclass1$.
\end{proof}

\noindent
Lemma \ref{lem:only-sf} reduces our task to proving that every star forest
that is not listed in Theorem \ref{thm:CR1}, that is, different from any of
\begin{equation}
  \label{eq:in-R1}
K_{1,s}\ (s\ge1),\,
2K_{1,1},\,
K_{1,2}+K_{1,1},\,
2K_{1,2},\,
K_{1,2}+2K_{1,1}
\end{equation}
does not belong to \rclass1.
Our proof of this fact sticks to the following scheme.
First, we will give a direct proof of non-membership for a small
amount of \emph{basic} star forests. Then we will establish
two \emph{derivation rules} based on some closure properties of \rclass1.
Finally, we will show that these derivation rules can be used,
for each star forest $F$ under consideration, to refute
the hypothesis $F\in\rclass1$ by deriving from it the membership
in \rclass1 of one of the basic star forests.

\begin{lemma}[Basic star forests]\label{lem:basic}
None of the star forests 
$K_{1,s}+K_{1,1}$ for any $s\ge3$, 
$K_{1,3}+K_{1,2}$,
$2K_{1,3}$, and 
$2K_{1,s}+K_{1,1}$ for any $s\ge1$
belongs to~\rclass1.
\end{lemma}

\begin{proof}
In order to prove that a graph $F$ is not in \rclass1,
one needs to exhibit \WL1-indistinguishable graphs $G$ and $H$
such that $G$ contains $F$ as a subgraph while $H$ does not. Below we provide 
such witnesses $G$ and $H$ for each basic star forest $F$ listed in the lemma;
see also Fig.~\ref{fig:GH-for-basic-sfs}.

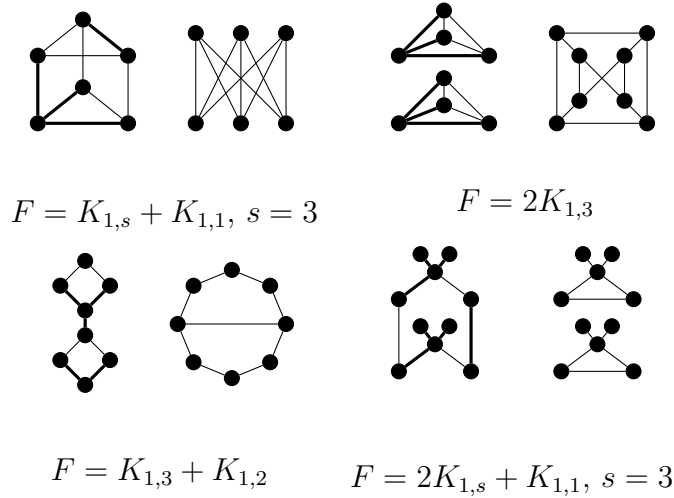
\begin{figure}[t]
  \centering
\begin{tikzpicture}[every node/.style={circle,draw,inner sep=2pt,fill=black},scale=.6]
  \begin{scope}
  \begin{scope}
   \path (0,0) node (a1) {}
    (2,0) node (a3) {}  edge[very thick] (a1)
      (1,.8) node (a4) {}  edge[very thick] (a1) edge (a3)

 (0,1.5) node (b1) {} edge[very thick] (a1)
    (2,1.5) node (b3) {}  edge (b1)  edge (a3)
      (1,2.3) node (b4) {}  edge (b1) edge[very thick] (b3) edge (a4);
 \end{scope}
  \begin{scope}[xshift=35mm]
    \path (0,0) node (a1) {}
      (1,0) node (a2) {} 
   (2,0) node (a3) {} 
    (0,2) node (b1) {} edge (a1) edge (a2) edge (a3)
      (1,2) node (b2) {} edge (a1) edge (a2) edge (a3)
     (2,2) node (b3) {} edge (a1) edge (a2) edge (a3);
 \end{scope}
\node[draw=none,fill=none,below] at (2.8,1.5) {$F=K_{1,s}+K_{1,1}$, $s=3$};
  \end{scope}

  \begin{scope}[xshift=80mm]
  \begin{scope}
    \path (0,0) node (a1) {}
      (1,.4) node (a2) {} edge[very thick] (a1)
   (2,0) node (a3) {}  edge[very thick] (a1) edge (a2)
      (1,1) node (a4) {}  edge[very thick] (a1) edge (a2) edge (a3)

 (0,1.5) node (b1) {}
      (1,1.9) node (b2) {}  edge[very thick] (b1)
   (2,1.5) node (b3) {}  edge[very thick] (b1) edge (b2)
      (1,2.5) node (b4) {}  edge[very thick] (b1) edge (b2) edge (b3);
 \end{scope}
  \begin{scope}[xshift=40mm,yshift=5mm]
    \path (0,0) node (a1) {}
      (1,0) node (a2) {}
   (1,1) node (a3) {}   edge (a1) edge (a2)
      (0,1) node (a4) {}  edge (a1) edge (a2)

  (-.5,-.5) node (b1) {} edge (a1)
      (1.5,-.5) node (b2) {} edge (a2)  edge (b1)
   (1.5,1.5) node (b3) {}  edge (a3) edge (b2)
      (-.5,1.5) node (b4) {}  edge (b1) edge (a4) edge (b3);
 \end{scope}
\node[draw=none,fill=none,below] at (2.8,-.1) {$F=2K_{1,3}$};
  \end{scope}

  \begin{scope}[xshift=5mm,yshift=-47mm]
  \begin{scope}[scale=.55,xshift=10mm]
   \path (0,1) node (a1) {}
    (1,2) node (a2) {}  edge[very thick] (a1)
      (0,3) node (a3) {}  edge (a2) 
      (-1,2) node (a4) {}  edge[very thick] (a1) edge (a3)

(0,0) node (b1) {} edge[very thick] (a1)
    (1,-1) node (b2) {}  edge (b1)
      (0,-2) node (b3) {}  edge[very thick] (b2) 
      (-1,-1) node (b4) {}  edge (b1) edge[very thick] (b3);
 \end{scope}
\newcommand{\ra}{1.2cm}
  \begin{scope}[xshift=38mm,yshift=2.5mm]
    \path (0:\ra) node (a1) {}
      (45:\ra) node (a2) {} edge (a1)
      (90:\ra) node (a3) {} edge (a2)
      (135:\ra) node (a4) {} edge (a3)
      (180:\ra) node (a5) {} edge (a4) edge (a1)
      (225:\ra) node (a6) {} edge (a5)
      (270:\ra) node (a7) {} edge (a6)
      (315:\ra) node (a8) {} edge (a7) edge (a1);
 \end{scope}
\node[draw=none,fill=none,below] at (2.2,-.5) {$F=K_{1,3}+K_{1,2}$};
  \end{scope}

  \begin{scope}[xshift=80mm,yshift=-55mm]
  \begin{scope}[scale=.8]
    \path (0,0) node (a1) {}
      (2,0) node (a2) {} 
   (1,.75) node (a3) {}  edge[very thick] (a1) edge (a2)
      ($(a3)+(-.4,.5)$) node (a4) {} edge[very thick] (a3)
      ($(a3)+(.4,.5)$) node (a5) {} edge[very thick] (a3)
($(a1)+(0,2)$) node (b1) {}  edge (a1)
 ($(a2)+(0,2)$) node (b2) {}  edge[very thick] (a2)
 ($(a3)+(0,2)$) node (b3) {}  edge[very thick] (b1) edge (b2)
 ($(a4)+(0,2)$) node (b4) {} edge[very thick] (b3)
 ($(a5)+(0,2)$) node (b5) {} edge[very thick] (b3);
\end{scope}
  \begin{scope}[scale=.8,xshift=45mm]
    \path (0,0) node (a1) {}
      (2,0) node (a2) {}  edge (a1)
   (1,.75) node (a3) {}  edge (a1) edge (a2)
      ($(a3)+(-.4,.5)$) node (a4) {} edge (a3)
      ($(a3)+(.4,.5)$) node (a5) {} edge (a3)
($(a1)+(0,2)$) node (b1) {}
 ($(a2)+(0,2)$) node (b2) {}   edge (b1) 
 ($(a3)+(0,2)$) node (b3) {}  edge (b1) edge (b2)
 ($(a4)+(0,2)$) node (b4) {} edge (b3)
 ($(a5)+(0,2)$) node (b5) {} edge (b3);
 \end{scope}
\node[draw=none,fill=none,below] at (2.5,1.3) {$F=2K_{1,s}+K_{1,1}$, $s=3$};
  \end{scope}
\end{tikzpicture}

\vspace{-15mm}

  \caption{$G$/$H$-certificates for each basic star forest $F$.} 
  \label{fig:GH-for-basic-sfs}
\end{figure}

\begin{description}
\item[$K_{1,s}+K_{1,1}$, $s\ge3$:]
$H=K_{s,s}$ \, and \, $G$ is obtained from $2K_s$ by adding a perfect matching between 
the two $K_s$ parts. 
\item[$2K_{1,3}$:]
$G=2K_4$ and $H$ is the \emph{Wagner graph} (or \emph{4-Möbius ladder}).
\item[$K_{1,3}+K_{1,2}$:]
$G$ is obtained from $2C_4$ by adding an edge between the two $C_4$ parts,
and $H$ is obtained from $C_8$ by adding an edge between two antipodal vertices of the 8-cycle in~$H$.
\item[$2K_{1,s}+K_{1,1}$, $s\ge1$:]
Both graphs $G$ and $H$ are obtained from $2K_{1,s+1}$ by adding two edges $e$. 
Let $a$ and $b$ be two leaves of the fist copy of $K_{1,s+1}$, 
and let $a'$ and $b'$ be two leaves of the other copy of $K_{1,s+1}$. 
Then $G$ additionally contains two edges $aa'$ and $bb'$,
whereas $H$ additionally contains two edges $ab$ and~$a'b'$.
\end{description}

\noindent
In the first two cases, the graphs $G$ and $H$ in each witness pair
are indistinguishable by color refinement as they are regular graphs
of the same degree with the same number of vertices. In the last two
cases, the \WL1-indistinguishability of $G$ and $H$ is easily seen
directly or by computing their stable partitions and applying
Lemma~\ref{lem:degree}.  
\end{proof}

\begin{lemma}[Derivation rules]\label{lem:rules}\hfill
  \begin{enumerate}[\bf 1.]
  \item 
If $K_{1,i_1}+\ldots+K_{1,i_s}+K_{1,i_{s+1}}\in\rclass1$, then $K_{1,i_1}+\ldots+K_{1,i_s}\in\rclass1$.
\item 
If $K_{1,i_1+1}+\ldots+K_{1,i_s+1}\in\rclass1$, then $K_{1,i_1}+\ldots+K_{1,i_s}\in\rclass1$.
  \end{enumerate}
\end{lemma}

\begin{proof}
  \textit{1.}
Suppose that $K_{1,i_1}+\ldots+K_{1,i_s}\notin\rclass1$.  Let $G$ and
$H$ be two graphs witnessing this, that is, $G\eqkkwl1H$ and $G$
contains this star forest while $H$ does not.  Then the graphs
$G+K_{1,i_{s+1}}$ and $H+K_{1,i_{s+1}}$, which are
\WL1-indistinguishable by Lemma \ref{lem:AABB}, witness that
$K_{1,i_1}+\ldots+K_{1,i_s}+K_{1,i_{s+1}}\notin\rclass1$.

  \textit{2.}
Suppose that $K_{1,i_1}+\ldots+K_{1,i_s}\notin\rclass1$ and this is witnessed by $G$ and $H$.
Given a graph $X$, let $X'$ denote the result of attaching a new degree-1 vertex $x'$
to each vertex $x$ of $X$ (thus, $v(X')=2v(X)$). Then the graphs $G'$ and $H'$ witness
that $K_{1,i_1+1}+\ldots+K_{1,i_s+1}\notin\rclass1$. 
Indeed, it is easy to see that $X$ contains $K_{1,i_1}+\ldots+K_{1,i_s}$
if and only if $X'$ contains $K_{1,i_1+1}+\ldots+K_{1,i_s+1}$ as a subgraph.
The equivalence $G'\eqkkwl1H'$ follows from the equivalence $G\eqkkwl1H$.
\end{proof}

\noindent
Now, let $F$ be a star forest not listed in \refeq{in-R1}.
Assume that $F\in\rclass1$. Lemma \ref{lem:rules} provides us
with two derivations rules:
\begin{itemize}
\item 
if a star forest $X$ is in \rclass1, then the result of removing
one connected component from $X$ is also in~\rclass1;
\item 
if a star forest $X$ is in \rclass1, then the result of cutting off
one leaf in each connected component of $X$ is also in~\rclass1.
\end{itemize}

\noindent
Note that, applying these derivation rules, $F$ can be reduced to one
of the basic star forests. By Lemma \ref{lem:basic}, we get a
contradiction, which completes the proof of Theorem~\ref{thm:CR1}.

\section{Weisfeiler-Leman invariance}\label{s:Ck}

The original algorithm described by Weisfeiler and Leman in \cite{WLe68},
which is nowadays more often referred to as the \emph{2-dimensional Weisfeiler-Leman
algorithm}, operates on the Cartesian square $V^2$ of the vertex set of an input graph $G$.
Initially it assigns each pair $(u,v)\in V^2$ one of three colors, namely
\emph{edge} if $u$ and $v$ are adjacent, \emph{nonedge} if $u\ne v$ and $u$ and $v$ are non-adjacent,
and \emph{loop} if $u=v$. Denote this coloring by $C^0$.
The coloring of $V^2$ is then refined step by step. The coloring after the $i$-th
refinement step is denoted by $C^i$ and is computed as
\begin{equation}
  \label{eq:refine-2}
  C^i(u,v)=C^{i-1}(u,v)\mid\mset{C^{i-1}(u,w)\mid C^{i-1}(w,v)}_{w\in V},
\end{equation}
where $\mset{}$ denotes the multiset
and $\mid$ denotes the string concatenation (an appropriate encoding is assumed).

The $k$-dimensional version of the algorithm, \kWL, operates on $V^k$. The initial
coloring of a tuple $(u_1,\ldots,u_k)$ encodes its equality type
and the isomorphism type of the subgraph of $G$ induced by
the vertices $u_1,\ldots,u_k$. The color refinement is performed similarly to \refeq{refine-2}.
For example, if $k=3$, then
$$
C^{i}(u_1,u_2,u_3)=C^{i-1}(u_1,u_2,u_3)\mid
\big\{\!\!\big\{ \hspace{0.5mm}C^{i-1}(w,u_2,u_3)
\mid C^{i-1}(u_1,w,u_3)\mid C^{i-1}(u_1,u_2,w) \hspace{0.5mm} \big\}\!\!\big\}_{w\in V}.  
$$
Generally, we write $\mathrm{WL}_k^r(G,u_1,\allowbreak \ldots,u_k)$ to denote the color of
the tuple $(u_1,\ldots,u_k)$ produced by the $k$-dimensional  Weisfeiler-Leman
algorithm after performing $r$ refinement steps. The length of $\alg krG{u_1,\ldots,u_k}$
grows exponentially as $r$ increases, which is remedied by renaming the tuple colors after each step
and retaining the corresponding color substitution tables. However, 
in our analysis of the algorithm we will use $\alg krG{u_1,\ldots,u_k}$ 
in its literal, iteratively defined meaning.

Let $\algg krG=\msetdef{\alg krG\baru}{\baru\in V^k}$ denote the color palette
observed on the input graph $G$ after $r$ refinement rounds.
We say that the $k$-dimensional  Weisfeiler-Leman algorithm \emph{distinguishes}
graphs $G$ and $H$ if $\algg krG\ne\algg krH$ after some number of rounds $r$.
The standard color stabilization argument shows that if $n$-vertex graphs $G$ and $H$
are distinguishable by \kWL, then they are distinguished
after $n^k$ refinement rounds at latest. If this does not happen, we say that
$G$ and $H$ are \emph{\kWL-equivalent} and write $G\eqkwl H$.

Obviously, isomorphic graphs are \kWL-equivalent for every $k$.
Recall also that any two strongly regular graphs with the same parameters
are \WL2-equivalent. The smallest pair of non-isomorphic strongly regular graphs
with the same parameters consists of the $4\times4$-rook's graph and the Shrikhande graph
(these graphs are depicted in Fig.~\ref{fig:shrikhande-rooks}).
The \WL2-equivalence of these graphs will be used several times below.

Note that the above description of the $k$-dimensional  Weisfeiler-Leman algorithm
and the \kWL-equivalence relation for $k\ge2$ are meaningful as well for
vertex-colored graphs (the initial coloring $C^0$ includes now also vertex colors).
We will need this more general framework only once, namely
in the proof of Theorem \ref{thm:one4clique} below.

\begin{theorem}[Dell, Grohe, and Rattan \cite{DellGR18}]\label{thm:DGR}
Let $\hhomo FG$ denote the number of homomorphisms from a graph $F$ to a graph $G$.
For each $F$ of treewidth $k$, the homomorphism count $\hhomo F\cdot$ is $\eqkwl$-invariant.
\end{theorem}

\begin{definition}
We define the \emph{homomorphism-hereditary treewidth} of a graph $F$,
denoted by $\htw F$, to be the maximum treewidth $\tw{F'}$ over all
homomorphic images $F'$ of~$F$.
\end{definition}

The following result follows directly from Theorem \ref{thm:DGR}
and the fact established by Lov{\'{a}}sz \cite[Section 5.2.3]{hombook}
that the subgraph count $\sub FG$ is expressible as a function of
the homomorphism counts $\hhomo{F'}G$ where $F'$ ranges over homomorphic
images of $F$ (see also \cite{CurticapeanDM17},
where algorithmic consequences of this relationship are explored).

\begin{corollary}\label{cor:C2}
\cclass k contains all $F$ with $\htw F\le k$.  
\end{corollary}

It is easy to see that $\htw F=1$ if and only if $F$ is a star graph
or the matching graph $2K_2$ (up to adding isolated vertices). Thus,
Theorem \ref{thm:CR1} implies that \cclass1 consists exactly of the
pattern graphs $F$ with $\htw F=1$.
We now characterize the class of graphs $F$ with $\htw F\le2$.

Given a graph $G$ and a partition $P$ of the vertex set $V(G)$, we define
the \emph{quotient graph} $G/P$ as follows. The vertices of  $G/P$ are
the elements of $P$, and $X\in P$ and $Y\in P$ are adjacent in $G/P$
if and only if $X\ne Y$ and there are vertices $x\in X$ and $y\in Y$
adjacent in $G$.

\begin{lemma}\label{lem:htw2}
$\htw F>2$ if and only
 if there is a partition $P$ of $V(F)$ such that $F/P\cong K_4$.  
\end{lemma}

\begin{proof}
Let us make two basic observations. First,
$H$ is a homomorphic image of $G$ if and only if there is a partition
$P$ of $V(G)$ into independent sets such that $H\cong G/P$.  
Second, $H$ is a minor of $G$ if and only there is 
a partition $P$ of $V(G)$ such that the graph $G[X]$ is connected
 for every $X\in P$ and $H$ is isomorphic to a subgraph of $G/P$. 

These observations imply the following fact, which is more general
than stated in the lemma.
Let $\mathcal S_k$ be the set of the minimal forbidden minors for 
the class of graphs with treewidth at most $k$. Note that, since the last class
of graphs is minor-closed, $\mathcal S_k$ exists and is finite by the Robertson–Seymour theorem.
Then $\htw F > k$ if and only if $V(F)$ admits a partition $P$ such that $G/P$
contains a subgraph isomorphic to a graph in~$\mathcal S_k$.

The lemma now follows from the well-known fact \cite[Chapter 12]{Diestel} that $\mathcal S_2=\{K_4\}$.
Note that, if $F/P$ contains $K_4$ as a subgraph, then $V(F)$ admits a partition $P'$ 
such that $F/P'$ is itself isomorphic to $K_4$ as the superfluous nodes of $F/P$ can be merged.  
\end{proof}

Whether or not $\htw F \leq 2$ is a necessary condition for the membership
of $F$ in \cclass2, is open.
We now show the equivalence of $F\in\cclass 2$ and $\htw F \leq 2$
for several standard graph sequences.

\begin{theorem}\label{thm:C2}
  \cclass2 contains 
  \begin{enumerate}[\bf 1.]
  \item
$K_2,2K_2,3K_2,4K_2,5K_2$ and no other matching graphs;
\item 
 $C_3,\ldots,C_7$ and no other cycle graphs;
\item 
 $P_1,\ldots,P_7$ and no other path graphs.
  \end{enumerate}
\end{theorem}

Theorem \ref{thm:C2} is related to some questions
that have earlier been discussed in the literature.
Beezer and Farrell \cite{BeezerF00} proved that the first five coefficients
of the matching polynomial of a strongly regular graph are determined
by its parameters. In other terms, if $G$ and $H$ are strongly regular
graphs with the same parameters (in fact, even distance-regular graphs
with the same intersection array), then $\sub{sK_2}G=\sub{sK_2}H$ for $s\le5$.
Part 1 of Theorem \ref{thm:C2} implies that this is true in a much more general
situation, namely when $G$ and $H$ are arbitrary $\WL2$-equivalent graphs.
Moreover, this cannot be extended to larger~$s$.

Fürer~\cite{Fuerer17} classified all $C_s$ for $s\leq 16$, except the $C_7$,
with respect to membership in \cclass2. Part 2 of Theorem \ref{thm:C2}
fills this gap and also shows that
the positive result for $C_7$ is optimal.

\begin{figure}[t]
\begin{center}
\begin{tikzpicture}[every node/.style={circle,fill,inner sep=2.0pt}]
\begin{scope}
 \graph[branch down=8mm,grow right=8mm,empty nodes] {
 r00 -- r01 -- r02 -- r03;
 r10 -- r11 -- r12 -- r13;
 r20 -- r21 -- r22 -- r23;
 r30 -- r31 -- r32 -- r33;
 r00 --[bend left=40] { r03};
 r10 --[bend left=40] { r13};
 r20 --[bend right=40] { r23};
 r30 --[bend right=40] { r33};
 r00 --[bend left=30,dashed] { r02}; r01 --[bend left=30,dashed] r03;
 r10 --[bend left=30,dashed] { r12}; r11 --[bend left=30,dashed] r13;
 r20 --[bend right=30,dashed] { r22}; r21 --[bend right=30,dashed] r23;
 r30 --[bend right=30,dashed] { r32}; r31 --[bend right=30,dashed] r33;
 r00 -- r10 -- r20 -- r30;
 r01 -- r11 -- r21 -- r31;
 r02 -- r12 -- r22 -- r32;
 r03 -- r13 -- r23 -- r33;
 r00 --[bend right=40] { r30};
 r01 --[bend right=40] { r31};
 r02 --[bend left=40] { r32};
 r03 --[bend left=40] { r33};
 r00 --[bend right=30,dashed] { r20}; r10 --[bend right=30,dashed] r30;
 r01 --[bend right=30,dashed] { r21}; r11 --[bend right=30,dashed] r31;
 r02 --[bend left=30,dashed] { r22}; r12 --[bend left=30,dashed] r32;
 r03 --[bend left=30,dashed] { r23}; r13 --[bend left=30,dashed] r33;
 };
\end{scope}

\begin{scope}[xshift=4cm]
 \graph[branch down=8mm,grow right=8mm,empty nodes] {
 s00 -- s01 -- s02 -- s03;
 s10 -- s11 -- s12 -- s13;
 s20 -- s21 -- s22 -- s23;
 s30 -- s31 -- s32 -- s33;
 s00 --[bend left=40] { s03};
 s10 --[bend left=40] { s13};
 s20 --[bend right=40] { s23};
 s30 --[bend right=40] { s33};
 s00 -- s10 -- s20 -- s30;
 s01 -- s11 -- s21 -- s31;
 s02 -- s12 -- s22 -- s32;
 s03 -- s13 -- s23 -- s33;
 s00 --[dashed] s11 --[dashed] s22 --[dashed] s33 --[bend right=20,looseness=0.75,dashed] s00;
 s01 --[dashed] s12 --[dashed] s23[dashed] --[dashed] s30 --[dashed] s01;
 s02 --[dashed] s13 --[dashed] s20 --[dashed] s31 --[dashed] s02;
 s03 --[dashed] s10 --[dashed] s21 --[dashed] s32 --[dashed] s03;
 s00 --[bend right=40] { s30};
 s01 --[bend right=40] { s31};
 s02 --[bend left=40] { s32};
 s03 --[bend left=40] { s33};
 };
\end{scope}

\begin{scope}[xshift=11cm,yshift=-1.1cm]
\path (0,0) node[draw=none,fill=none] (a1) {$\begin{array}{|c|r|r|}\hline
n & \sub {P_n}G & \sub {P_n}H \\
\hline
 8 &   275616&   274560\\
 9 &   880128&   877440\\
10 &  2506752&  2512512\\
11 &  6239232&  6283392\\
12 & 13189248& 13293696\\
13 & 22631040& 22754688\\
14 & 29376000& 29457408\\
15 & 25532928& 25560576\\
16 & 11197440& 11115264\\\hline
\end{array}$};
\end{scope}

\end{tikzpicture}
\end{center}\vspace{-10mm}
\caption{The 4x4 rook's graph $G$ and the Shrikhande graph $H$
(some edges are dashed just to ensure readability). The table shows
the counts of $P_n$, $8\le n\le16$, in $G$ and $H$.}
\label{fig:shrikhande-rooks}
\end{figure}

\begin{proof}
\textit{1.}
Lemma \ref{lem:htw2} makes it obvious that $\htw{sK_2}\le2$ exactly for $s\le5$.
This gives the positive part by Corollary \ref{cor:C2}.
It remains to prove that $sK_2\notin\cclass2$ for all $s\ge6$.
For $6K_2$, let $G$ be the $4\times4$-rook's graph and $H$ be the Shrikhande graph;
see Fig.~\ref{fig:shrikhande-rooks}.
Being strongly regular graphs with the same parameters $(16,6,2,2)$,
$G$ and $H$ are \WL2-equivalent. As calculated in \cite{BeezerF00},
$\sub{6K_2}G=96000$ while $\sub{6K_2}H=95872$, which certifies the
non-membership of $6K_2$ in \cclass2. In order to extend this to $sK_2$ for $s>6$,
note that
$$
\sub{(s+1)K_2}{G+K_2}=\sub{(s+1)K_2}G+\sub{sK_2}G.
$$
This equality implies that if $\sub{sK_2}G \ne \sub{sK_2}H$, then
it holds also one of the inequalities
$\sub{(s+1)K_2}G \ne \sub{(s+1)K_2}H$ or
$\sub{(s+1)K_2}{G+K_2} \ne \sub{(s+1)K_2}{H+K_2}$.
Thus, if a pair $G,H$ is a certificate for $sK_2\notin\cclass2$,
then  $(s+1)K_2\notin\cclass2$ is certified by the same pair $G,H$
or by the pair $G+K_2,H+K_2$. In the latter case we need to remark
that $G+K_2\eqkkwl2 H+K_2$ whenever $G\eqkkwl2 H$.

\textit{2 and 3.}
We have $\htw{C_s}\le2$ if $s\le7$. 
We can use Lemma \ref{lem:htw2} to show this.
For $s\le5$ this is obvious because $K_4$ has 6 edges.
This is easy to see also for $s=6$: we cannot get $K_4$ by merging just two vertices,
while merging more than two vertices results in loss of one edge.
Let $s=7$. The argument for $s=6$ shows that $K_4$ cannot be obtained
from $C_7$ if at least one edge is contracted.
The assumption $C_7/P\cong K_4$ would, therefore, mean that
$K_4$ has a closed walk that uses one edge twice and every other edge once. 
Such a walk contains an Eulerian trial in
$K_4$, which is impossible because $K_4$ has all four vertices of degree~3.

Corollary \ref{cor:C2}, therefore, implies that \cclass2 contains all cycle graphs $C_s$
up to $s=7$. It contains also all paths $P_s$ up to $s=7$, as 
the class of graphs $\setdef F{\htw F\le2}$ is closed under taking subgraphs,
which is easily seen from Lemma~\ref{lem:htw2}.

% \begin{table}
% \begin{center}
% \begin{tabular}{|r|r|r|}
% $n$ & $\sub {P_n}G$ & $\sub {P_n}H$ \\
% \hline
% 8 & 274560 &  275616 \\
% 9 & 877440 &  880128\\
% 10 & 2512512 & 2506752\\
% 11 & 6283392 & 6239232\\
% 12 & 13293696 & 13189248\\
% 13 & 22754688 & 22631040\\
% 14 & 29457408 & 29376000\\
% 15 & 25560576 & 25532928\\
% 16 & 11115264 & 11197440
% \end{tabular}
% \end{center}
% \caption{Path counts in Shrikhande graph $G$ and 4x4 rook's graph $H$.}
% \label{tab:pathcounts}
% \end{table}

In order to obtain the negative part, we again use the Shrikhande
graph $G$ and the $4\times4$ rook's graph $H$.  For $s\leq 16$ see the
table in Fig.~\ref{fig:shrikhande-rooks} for $P_s$ and \cite{Fuerer17}
for $C_s$.  For $s > 16$ construct the graphs $G_s$ and $H_s$ by
adding a vertex-disjoint path $P_{s-16}$ to $G$ and $H$ respectively
and by connecting both end vertices of this path to all original
vertices of $G$ and $H$. Then
$$
\sub {C_s}{G_s}=\sub {P_{16}}{G}\neq \sub {P_{16}}{H} = \sub {C_s}{H_s},
$$ 
while still $G_s \eqkkwl 2 H_s$.

For paths, we use almost the same construction of $G_s$ and $H_s$, where we connect only
one end vertex of $P_{s-16}$ to the original graph. Then
$$
\sub {P_s}{G_s}=2\,\sub {P_{16}}{G}\neq 2\,\sub {P_{16}}{H} = \sub {P_s}{H_s},
$$ 
and the pair $G_s,H_s$ certifies that $P_s\notin\cclass2$.
This works for all $s > 17$.
If $s = 17$, we construct graphs $H_{17}$ and $G_{17}$
by adding a new neighbor of degree one to each vertex in $G$ and $H$.
Then 
$$\sub {P_{17}}{G_{17}}=2\,\sub {P_{16}}{G}+\sub {P_{15}}{G}$$
and analogously for $H$ and $H_{17}$. It remains to use the table in
Fig.~\ref{fig:shrikhande-rooks} to see that this sum is different for
$G$ and~$H$.  
\end{proof}

Part 2 of Theorem \ref{thm:C2} can be generalized as follows.
Recall that $g(G)$ denotes the girth of a graph~$G$.

\begin{theorem}\label{thm:g}
Suppose that $G\eqkkwl2H$. Then

\begin{enumerate}[\bf 1.]
\item
$g(G)=g(H)$.
\item 
$\sub{C_s}G=\sub{C_s}H$ for each $3\le s\le 2\,g(G)+1$.
\end{enumerate}
\end{theorem}

\begin{proof}
\textit{1.}
The proof uses the logical characterization of the $\eqkwl$-equivalence in~\cite{CaiFI92}.
According to this characterization, $G\eqkwl H$ if and only if $G$ and $H$ satisfy the same
sentences in the first-order $(k+1)$-variable logic with counting quantifiers $\E^{\ge t}$,
where an expression $\E^{\ge t}x\,\Phi(x)$ for any integer $t$ means that there are at least $t$ 
vertices $x$ with property $\Phi(x)$.

Assume that $g(G)<g(H)$ and show that then $G\not\eqkkwl2H$. It is enough
to show that $G$ and $H$ are distinguishable in 3-variable logic with
counting quantifiers.

\Case1{$g(G)$ is odd.}
In this case, $G$ and $H$ are distinguishable even in the standard 3-variable logic (with
quantifiers $\E$ and $\A$ only). As it is well known \cite{dcbook}, 
two graphs $G$ and $H$ are distinguishable in first-order $k$-variable logic
if and only if Spoiler has a winning strategy in the \emph{$k$-pebble Ehrenfeucht-Fra{\"\i}ss{\'e} game}
on $G$ and $H$. In the 3-pebble game, the players \emph{Spoiler} and \emph{Duplicator}
have equal sets of $3$ pebbles $\{a,b,c\}$.
In each round, Spoiler takes a pebble and puts it on a vertex in $G$ or in $H$;
then Duplicator has to put her copy of this pebble on a vertex
of the other graph.
Duplicator's objective is to ensure that the pebbling determines a partial
isomorphism between $G$ and $H$ after each round; when she fails, she immediately loses.

Spoiler wins the game as follows.  Let $C$ be a cycle of length $g(G)$
in $G$. In the first three rounds, Spoiler pebbles a 3-path along $C$
by his pebbles $a$, $b$, and $c$ in this order.  Then, keeping the
pebble $a$ fixed, Spoiler moves the pebbles $b$ and $c$, in turns,
around $C$ so that the two pebbled vertices are always adjacent. In
the end, there arises a pebbled $acb$-path, which is impossible
in~$H$.

\Case2{$g(G)$ is even.} Let $g(G)=2m$.
Consider the following statement in the 3-variable logic with counting quantifiers:
$$
\E x\E y\of{
\mathit{dist}(x,y)=m
\und
\E^{\ge2}z(z\sim y \und \mathit{dist}(z,x)=m-1)
},
$$
where $\mathit{dist}(x,y)=m$ is a 3-variable formula expressing the fact
that the distance between vertices $x$ and $y$ is equal to $m$.
This statement is true on $G$ and false on~$H$.

\smallskip

\textit{2.}
The proof of this part is based on the result by Dell, Grohe, and Rattan stated
above as Theorem \ref{thm:DGR} and Lov{\'{a}}sz' result \cite[Section 5.2.3]{hombook}
on the expressibility of $\sub FG$ through
the homomorphism counts $\hhomo{F'}G$ for homomorphic images $F'$ of $F$.
By these results, it suffices to prove that, if $s\le 2\,g(G)+1$ and $h$ is a homorphism from $C_s$
to $G$, then the subgraph $h(C_s)$ of $G$ has treewidth at most 2.
Assume, to the contrary, that $h(C_s)$ has treewidth more than 2 or,
equivalently, $h(C_s)$ contains $K_4$ as a minor. Since $K_4$ has maximum degree 3,
$h(C_s)$ contains $K_4$ even as a topological minor \cite[Section 1.7]{Diestel}.
Let $M$ be a subgraph of $h(C_s)$ that is a subdivision of $K_4$.
Obviously, $s\ge e(h(C_s))\ge e(M)$. Moreover, $s\ge e(M)+2$.
Indeed, the homomorphism $h$ determines a walk of length $s$ via
all edges of the graph $h(C_s)$. By cloning the edges traversed more than once,
$h(C_s)$ can be seen as an Eulerian multigraph with $s$ edges.
Since $M$ has four vertices of degree 3, any extension of $M$ to
such a multigraph requires adding at least 2 edges.
Thus, $s\ge e(M)+2$. Note that $M$ is formed by six paths corresponding 
to the edges of $K_4$. Moreover, $M$ has four cycles, each cycle consists of three paths, 
and each of the six paths appears in two of the cycles. It follows that
$2\,e(M)\ge4\,g(G)$. Therefore, $s\ge2\,g(G)+2$, yielding a contradic\-tion.
\end{proof}

Moreover, Part 2 of Theorem \ref{thm:C2} admits a qualitative strengthening:
It turns out that even \rclass2 contains only finitely many cycle graphs $C_s$.
In fact, a much stronger fact is true.

\begin{theorem}\label{thm:Rk-cycles}
  For each $k$, the class \rclass k contains only finitely many cycle graphs~$C_s$.
\end{theorem}

\noindent
The proof of Theorem \ref{thm:Rk-cycles} follows a powerful approach suggested by Dawar in~\cite{Dawar98}
to prove that the graph property of containing a Hamiltonian cycle
is not $\eqkwl$-invariant for any $k$.
This fact alone immediately implies that, whatever $k$ is, \rclass k cannot contain
all cycle graphs~$C_s$. An additional effort is needed to show that no \rclass k
can contain infinitely many~$C_s$. Specifically,
Theorem \ref{thm:Rk-cycles} is a direct consequence of the following two facts.

\begin{lemma}
\mbox{}

\begin{enumerate}[\bf 1.]
\item 
No class \rclass k contains all path graphs, that is,
for very $k$ there is $t$ such that $P_t\notin\rclass k$.
\item 
If $P_t\notin\rclass k$, then $C_s\notin\rclass k$ for all $s>t$.
\end{enumerate}
\end{lemma}

\begin{proof}
\textit{1.}
We begin with description of the main idea of Dawar's method.
The Graph Isomorphism problem (\probl{GI}) is the recognition problem
for the set of all pairs of isomorphic graphs.
We can encode \probl{GI} as a class of relational structures
over vocabulary $\langle V_1,V_2,E\rangle$, where $V_1$ and $V_2$ are unary relations
describing two vertex sets and $E$ is a binary adjacency relation (over $V_1\cup V_2$).
Then \probl{GI} consists of those structures where $V_1$ and $V_2$ are disjoint
and the graphs $(V_1,E)$ and $(V_2,E)$ are isomorphic. The starting
point of the method is observing that \probl{GI} is not
$\eqkwl$-invariant for any $k$. This follows from the seminal work
by Cai, Fürer, and Immerman \cite{CaiFI92}, who constructed, for each $k$,
a pair of non-isomorphic graphs $G$ and $H$ such that $G\eqkwl H$: Indeed,
$G+G\in\text{\probl{GI}}$ and $G+H\notin\text{\probl{GI}}$,
and $G+G\eqkwl G+H$.

Suppose now that we have two classes of relational structures $\mathcal{C}_1$
and $\mathcal{C}_2$ and know that $\mathcal{C}_1$ is not $\eqkwl$-invariant for any $k$.
We can derive the same fact for $\mathcal{C}_2$ by showing a \emph{first-order
reduction} from $\mathcal{C}_1$ to $\mathcal{C}_2$, that is, a function
$f$ such that $f(A)\in\mathcal{C}_2$ iff $A\in\mathcal{C}_1$ for any structure $A$
in the vocabulary of $\mathcal{C}_1$, where
$f(A)$ is a structure in the vocabulary of $\mathcal{C}_2$ whose relations
are relations in the universe of $A$ and are definable by first-order
formulas over the vocabulary of $\mathcal{C}_1$; see \cite{dcbook} for details.

There is a first-order reduction from \probl{GI} to the \probl{Satisfiability}
problem (where for CNFs we assume a standard encoding as relational structures);
see, e.g.,~\cite{Toran13}.
A first-order reduction from \probl{Satisfiability} to \probl{Hamiltonian Cycle}
is described by Dahlhaus~\cite{Dahlhaus84}. 
We now describe a first-order reduction from \probl{Hamiltonian Cycle} to 
the problem \probl{Long Path}, which we define as the problem of recognizing
whether a given $N$-vertex graph contains a path of length at least $\frac34\,N+2$.

To this end, we modify a standard reduction from \probl{Hamiltonian Cycle} to \probl{Hamiltonian Path}
(which itself is not first-order as it requires selection of a single vertex from the
vertex set of a given graph). Specifically, suppose we are given a graph $G$
with $n$ vertices. We expand $G$ to a graph $G'$ with $8n$ vertices as follows.
For each vertex $v$ of $G$, we create its clone $v'$ with the same adjacency
to the other vertices of $G$ (and their clones). Next, we connect $v$ and $v'$
by a path $v v_1 v_2 v_3 v_4 v'$ via four new vertices $v_1,v_2,v_3,v_4$. 
Finally, for each $v\in V(G)$, we add a new neighbor $u$ to $v_2$ and
 a new neighbor $u'$ to $v_3$. In the resulting graph $G'$, $u$ and $u'$ have degree 1. 
It remains to notice that $G$ has a Hamiltonian cycle if and only if
$G'$ has a path of length $6n+2$.
The vertex set and the adjacency relation of $G'$ can easily be defined
by first order formulas in terms of the vertex set and the adjacency relation of~$G$.

Composing the aforementioned reductions, we obtain a first-order
reduction from \probl{GI} to \probl{Long Path} and conclude
that the last problem is not $\eqkwl$-invariant for any $k$.
It remains to note that the existence of a fixed $k$ for which
$\rclass k$ contains all path graphs, would imply the $\eqkwl$-invariance
of \probl{Long Path} for this~$k$.

\textit{2.}
Consider a pair of graphs $G,H$ certifying that $P_t\notin\rclass k$, that is,
$G\eqkwl H$ and $G$ contains $P_t$ while $H$ does not.
Without loss of generality, we can suppose that $G$ and $H$ have no isolated vertices.
Like in the proof of Part 2 of Theorem \ref{thm:C2}, we construct the graph $G_s$ by adding
a vertex-disjoint path $P_{s-t}$ to $G$ connecting both 
end vertices of this path to all vertices of $G$. The graph $H_s$ is obtained similarly from $H$.
Note that $G_s\eqkwl H_s$ and that $G_s$ contains $C_s$ while $H_s$ does not.
\end{proof}

It is easy to see that $K_3\in\rclass2$ (this is also a formal consequence of
Part 2 of Theorem \ref{thm:C2}).  Using the pair $G,H$
consisting of the $4\times4$-rook's graph and the Shrikhande graph,
Fürer~\cite{Fuerer17} proved that the complete graph $K_4$ is not in
\rclass2.  By padding $G$ and $H$ with new $s-4$ universal vertices,
we see that \rclass2 contains $K_s$ if and only if $s\le3$. Fürer's
result on the non-membership of $K_4$ in \rclass2 admits the following
generalization.

\begin{theorem}\label{thm:one4clique}
No graph containing a unique 4-clique can be in $\rclass 2$.
\end{theorem}

Given a graph $R$, we define a corresponding vertex-colored graph
$R^*$, whose vertices are colored using four colors $1,2,3,4$, as
follows:
\begin{itemize}
\item Each vertex $v$ of $R$ is replaced by four clones
  $v_1,v_2,v_3,v_4$, where $v_i$ has color $i$;
\item For vertices $v$ and $u$ of $R$, their clones $v_i$ and $u_j$
  are adjacent in $R^*$ if and only if $u$ and $v$ are adjacent in $R$
  and $i\ne j$.
\end{itemize}

This transformation, which we require for the proof of
Theorem~\ref{thm:one4clique}, is based on a reduction from the
parametrized $k$-CLIQUE problem to its ``colorful version'' by Fellows
et al.~\cite[Lemma~1]{FellowsHRV09}.

\begin{lemma}\label{lem:fellows}\hfill
  \begin{enumerate}[\bf 1.]
  \item 
$R$ contains a 4-clique if and only if $R^*$ contains a 4-clique.
Moreover, the vertices of any 4-clique in $R^*$ have pairwise different colors.
\item 
If $R \eqkkwl 2 S$, then $R^* \eqkkwl 2 S^*$.
  \end{enumerate}
\end{lemma}

\begin{proof}
  Part 1 is easy. To prove Part 2, we use the fact \cite{Hella96} that $G \eqkkwl 2 H$ if and only if
Duplicator has a winning strategy in the 3-pebble \emph{Hella's bijection game} on $G$ and $H$. 

Like the Ehrenfeucht-Fra{\"\i}ss{\'e} game that we used in the proof of Theorem
\ref{thm:g}, the bijection game is played by two players, Spoiler and
Duplicator, to whom we will refer as \emph{he} and \emph{she}
respectively. Let $p_1,p_2,p_3$ be the three distinct pebbles. There
are two copies of each pebble $p_i$. In one round of the game, Spoiler
puts one of the pebbles $p_i$ on a vertex in $G$ and its copy on a
vertex in $H$. When $p_i$ is on the board, $x_i$ denotes the vertex
pebbled by $p_i$ in $G$, and $y_i$ denotes the vertex pebbled by the
copy of $p_i$ in $H$. The pebbles can change their positions during
the game and, thus, the values of $x_i$ and $y_i$ can be different in
different rounds.  More specifically, a round is played as follows:
\begin{itemize}
\item 
Spoiler chooses $i\in\{1,2,3\}$;
\item 
Duplicator responds with a bijection $f\function{V(G)}{V(H)}$ having the property that
$f(x_j)=y_j$ for all $j\ne i$ such that $p_j$ is on the board;
\item Spoiler chooses a vertex $x$ in $G$ and puts $p_i$ on $x$ and
  its copy on $f(x)$ (this move reassigns $x_i$ to vertex $x$ and
  $y_i$ to vertex $f(x)$).
\end{itemize}

Duplicator's objective is to keep the map $x_i\mapsto y_i$ a partial
isomorphism during the play. Spoiler wins if the Duplicator fails. If
$G$ and $H$ are vertex-colored graphs, then the Duplicator has to keep
the map $x_i\mapsto y_i$ a color-preserving partial isomorphism. The
description of the bijection game is complete.

The assumption $R \eqkkwl 2 S$ implies that Duplicator has a winning
strategy in the 3-pebble bijection game on $R$ and $S$.  She can
transform this strategy to the game on graphs $R^*$ and $S^*$.  Define
a projection map $\lambda\function{V(R^*)\cup V(S^*)}{V(R)\cup V(S)}$
as follows: If a vertex $w\in V(R^*)\cup V(S^*)$ is a clone of a
vertex $u\in V(R)\cup V(S)$, then $\lambda(w)=u$.  Duplicator
simulates a round of the game on $R$ and $S$ by assuming that
\begin{itemize}
\item 
Spoiler chooses the pebble with index $i\in\{1,2,3\}$ in the
  simulated game on $R$ and $S$ whenever he does it in the real game
  on $R^*$ and~$S^*$;
\item 
Spoiler chooses the vertex $\lambda(w)$ in $R$ whenever he
  chooses a vertex $w$ in~$R^*$.
\end{itemize}

Whenever Duplicator's strategy in the simulated game on $R$ and $S$
yields a bijection $f\function{V(R)}{V(S)}$, in the real game on $R^*$
and $S^*$ Duplicator responds with the bijection
$f^*\function{V(R^*)}{V(S^*)}$ taking each clone of a vertex $v\in
V(R)$ to the clone of $f(v)$ that has the same color. This completes
description of Duplicator's strategy for the game on $R^*$ and~$S^*$.

Note that, whenever $x_i\in V(R^*)$ and $y_i\in V(S^*)$ are pebbled by
$p_i$, then $\lambda(x_i)\in V(R)$ and $\lambda(y_i)\in V(S)$ are
pebbled by $p_i$ in the simulated game.  This, along with the facts
that Duplicator always succeeds in the simulated game and $f^*$ always
preserves the vertex colors, readily implies that Duplicator succeeds
in each round of the game on $R^*$ and $S^*$. Thus, she has a winning
strategy in this game, and we conclude that $R^* \eqkkwl 2 S^*$.
\end{proof}

\begin{proof}[Proof of Theorem \ref{thm:one4clique}]
Let $R$ be the $4\times4$-rook's graph and $S$ be the Shrikhande graph.
Recall that $R$ contains a 4-clique, while $S$ does not.
Consider now $G=R^*$ and $H=S^*$. By Lemma \ref{lem:fellows},
$G$ contains a 4-clique, $H$ does not, and $G \eqkkwl 2 H$.
Let $c_G\function{V(G)}{\{1,2,3,4\}}$ and $c_H\function{V(H)}{\{1,2,3,4\}}$
denote the vertex colorings of $G$ and $H$ respectively.
 
Now, let $F$ be a graph that contains exactly one $K_4$.  Suppose that
$V(F)=\Set{1,\dots,l}$ and $F[\Set{1,2,3,4}]\cong K_4$.  Denote
$F'=F[\Set{5,\dots,l}]$. We define a graph $G'$ as $V(G')=V(G)\cup
V(F')$ and
$$
E(G')=E(G)\cup E(F') \cup 
\setdef{\{u,v\}}{u\in V(G),\,v\in V(F'),\,\{c_G(u),v\}\in E(F)}.   
$$
In other words, each of the vertices 1, 2, 3, and 4 of $F$ is cloned to 16
copies with the same adjacency to the other vertices.  Further, the
set of the 64 clones, whose names 1, 2, 3, 4 are now regarded as
colors, is endowed with edges to create a copy of $G$. The graph $H'$
is defined similarly.

Note that $H'$ does not contain $F$ or even any copy of $K_4$.  Indeed,
$K_4$ appears neither in $H$ nor in $F'$, and any copy $K$ of $K_4$ in
$H'$ with an edge between $V(H)$ and $V(F')$ would give rise to one
more copy of $K_4$ in $F$ (note that $K$ can use only differently
colored vertices from $H$ because each color class of $H$ is an
independent set).  On the other hand, $G'$ contains $F$ as a
subgraph. Indeed, $G$ contains a $4$-clique with colors $1,2,3,4$,
which completes the $F'$ fragment of $G'$ to a copy of~$F$.

It remains to prove that $G' \eqkkwl 2 H'$.  It suffices to prove that
Duplicator has a winning strategy in the 3-pebble bijection game on
$G'$ and $H'$.  Since $G \eqkkwl 2 H$, Duplicator has a winning
strategy in the 3-pebble bijection game on $G$ and $H$.  She can win
the game on $G'$ and $H'$ by simulating the game on $G$ and $H$ as
follows.  She assumes that

\begin{itemize}
\item Spoiler chooses the pebble with index $i\in\{1,2,3\}$ in the
  simulated game on $G$ and $H$ whenever he does it in the real game
  on $G'$ and~$H'$;
\item Spoiler chooses a vertex $x$ in $G$ whenever he chooses this
  vertex in $G'$ (recall that $V(G')=V(G)\cup V(F')$).
\end{itemize}

Whenever Duplicator's strategy in the simulated game on $G$ and $H$
yields a bijection $f\function{V(G)}{V(H)}$, in the real game on $G'$
and $H'$ Duplicator responds with the bijection
$f'\function{V(G')}{V(H')}$ that coincides with $f$ on $V(G)$ and is
the identity map on $V(F')$.  The bijection $f'$ does not change if
Spoiler chooses a vertex $x$ in $V(F')\subset V(G')$.

In order to check that this strategy is winning for Duplicator,
consider the vertices pebbled in $G'$ and $H'$ by $p_i$ and $p_j$ for
any $i,j\in\{1,2,3\}$. Without loss of generality, suppose that $i=1$
and $j=2$. According to our notation, $p_1$ occupies vertices $x_1\in
V(G')$ and $y_1\in V(H')$, and $p_2$ occupies vertices $x_2\in V(G')$
and $y_2\in V(H')$. Duplicator's strategy ensures that $x_i\in V(F')$
exactly when $y_i\in V(F')$ for each $i=1,2$. If both $x_1$ and $x_2$
are in $V(G)$, then both $y_1$ and $y_2$ are in $V(H)$ and are
adjacent if and only if $x_1$ and $x_2$ are adjacent. The last
condition is true because the vertices in $V(G)\cup V(H)$ are pebbled
according to Duplicator's winning strategy in the game on $G$ and $H$.
If both $x_1$ and $x_2$ are in $V(F')$, then $y_1$ and $y_2$ is the
identical vertex pair in the graph $F$, and the adjacency relation is
preserved by trivial reasons.  Finally, suppose that $x_1\in V(G)$
while $x_2\in V(F')$ and, hence, $y_1\in V(H)$ and $y_2\in V(F')$.
Since $x_1$ and $y_1$ were pebbled according to Duplicator's winning
strategy in the game on $G$ and $H$, they have the same
color. Moreover, $x_2$ and $y_2$ are identical vertices in $F$.  It
follows by the construction of $G'$ and $H'$ that $y_1$ and $y_2$ are
adjacent in $H'$ if and only if $x_1$ and $x_2$ are adjacent in~$G'$.
\end{proof}

\section{Concluding discussion}

An intriguing open problem is whether Corollary \ref{cor:C2} yields a
complete description of the class \cclass k. Our Theorem \ref{thm:CR1}
gives an affirmative answer in the one-dimensional case.  Moreover,
this theorem gives a complete description of the class \rclass1. The
class \rclass2 remains a mystery. For example, it contains either
finitely many matching graphs $sK_2$ or all of them, and we currently
do not know which of these is true. In other words, is the matching
number preserved by $\eqkkwl2$-equivalence? Note that non-isomorphic
strongly regular graphs with the same parameters cannot yield
counterexamples to this. The \emph{Brouwer-Haemers conjecture} states
that every connected strongly regular graph is Hamiltonian except the
Petersen graph, and Pyber \cite{Pyber14} has shown there are at most
finitely many exceptions to this conjecture. Since the Petersen graph
has a perfect matching, it is therefore quite plausible that every
connected strongly regular graph has an (almost) perfect matching.

By Corollary \ref{cor:C2}, the subgraph count $\sub FG$ is
\kWL-invariant for $k=\htw F$. Interestingly, the parameter $\htw F$ appears in a result by
Curticapean, Dell, and Marx \cite{CurticapeanDM17} who show that $\sub
FG$ is computable in time $e(F)^{O(e(F))}\cdot v(G)^{\htw F +1}$.
An interesting area is to explore connections between
\kWL-invariance and algorithmics, which are hinted by this apparent
coincidence. 

Which \emph{induced} subgraphs and their counts are $\kWL$-invariant
for different $k$ deserves study. We note that the induced subgraph
counts have been studied in the context of finite model theory by
Kreutzer and Schweikardt \cite{KreutzerS14}.

%% to run bibtex:
% \bibliographystyle{abbrv}
% \bibliography{wl}

\begin{thebibliography}{10}

\bibitem{AndersonDH15}
M.~Anderson, A.~Dawar, and B.~Holm.
\newblock Solving linear programs without breaking abstractions.
\newblock {\em J. {ACM}}, 62(6):48:1--48:26, 2015.

\bibitem{ArvindKRV17}
V.~Arvind, J.~K{\"o}bler, G.~Rattan, and O.~Verbitsky.
\newblock Graph isomorphism, color refinement, and compactness.
\newblock {\em Computational Complexity}, 26(3):627--685, 2017.

\bibitem{AtseriasBD09}
A.~Atserias, A.~A. Bulatov, and A.~Dawar.
\newblock Affine systems of equations and counting infinitary logic.
\newblock {\em Theor. Comput. Sci.}, 410(18):1666--1683, 2009.

\bibitem{Babai16}
L.~Babai.
\newblock Graph isomorphism in quasipolynomial time.
\newblock In {\em Proceedings of the 48th Annual {ACM} Symposium on Theory of
  Computing ({STOC}'16)}, pages 684--697, 2016.

\bibitem{BabaiES80}
L.~Babai, P.~Erd\H{o}s, and S.~M. Selkow.
\newblock Random graph isomorphism.
\newblock {\em SIAM Journal on Computing}, 9(3):628--635, 1980.

\bibitem{BeezerF00}
R.~A. Beezer and E.~J. Farrell.
\newblock The matching polynomial of a distance-regular graph.
\newblock {\em Int. J. Math. Math. Sci.}, 23(2):89--97, 2000.

\bibitem{Bollobas-b}
B.~Bollob{\'a}s.
\newblock {\em Random graphs}, volume~73 of {\em Cambridge Studies in Advanced
  Mathematics}.
\newblock Cambridge University Press, Cambridge, second edition, 2001.

\bibitem{BoseM59}
R.~C. Bose and D.~M. Mesner.
\newblock On linear associative algebras corresponding to association schemes
  of partially balanced designs.
\newblock {\em Ann. Math. Statist.}, 30:21--38, 1959.

\bibitem{CaiFI92}
J.~Cai, M.~F{\"u}rer, and N.~Immerman.
\newblock An optimal lower bound on the number of variables for graph
  identifications.
\newblock {\em Combinatorica}, 12(4):389--410, 1992.

\bibitem{CurticapeanDM17}
R.~Curticapean, H.~Dell, and D.~Marx.
\newblock Homomorphisms are a good basis for counting small subgraphs.
\newblock In {\em Proceedings of the 49th Annual {ACM} {SIGACT} Symposium on
  Theory of Computing (STOC'17)}, pages 210--223. {ACM}, 2017.

\bibitem{Dahlhaus84}
E.~Dahlhaus.
\newblock Reduction to {NP}-complete problems by interpretations.
\newblock In E.~B{\"{o}}rger, G.~Hasenjaeger, and D.~R{\"{o}}dding, editors,
  {\em Logic and Machines: Decision Problems and Complexity}, volume 171 of
  {\em Lecture Notes in Computer Science}, pages 357--365. Springer, 1984.

\bibitem{Dawar98}
A.~Dawar.
\newblock A restricted second order logic for finite structures.
\newblock {\em Inf. Comput.}, 143(2):154--174, 1998.

\bibitem{Dawar15}
A.~Dawar.
\newblock The nature and power of fixed-point logic with counting.
\newblock {\em {SIGLOG} News}, 2(1):8--21, 2015.

\bibitem{DawarSZ17}
A.~Dawar, S.~Severini, and O.~Zapata.
\newblock Pebble games and cospectral graphs.
\newblock {\em Electronic Notes in Discrete Mathematics}, 61:323--329, 2017.

\bibitem{DellGR18}
H.~Dell, M.~Grohe, and G.~Rattan.
\newblock Lov{\'{a}}sz meets {W}eisfeiler and {L}eman.
\newblock In {\em 45th International Colloquium on Automata, Languages, and
  Programming ({ICALP}'18)}, volume 107 of {\em LIPIcs}, pages 40:1--40:14.
  Schloss Dagstuhl -- Leibniz-Zentrum fuer Informatik, 2018.

\bibitem{Diestel}
R.~Diestel.
\newblock {\em Graph theory}.
\newblock New York, NY: Springer, 2000.

\bibitem{FarrellGC91}
E.~J. Farrell, J.~M. Guo, and G.~M. Constantine.
\newblock On matching coefficients.
\newblock {\em Discrete Mathematics}, 89(2):203--210, 1991.

\bibitem{FellowsHRV09}
M.~R. Fellows, D.~Hermelin, F.~A. Rosamond, and S.~Vialette.
\newblock On the parameterized complexity of multiple-interval graph problems.
\newblock {\em Theor. Comput. Sci.}, 410(1):53--61, 2009.

\bibitem{Fuerer17}
M.~F{\"u}rer.
\newblock On the combinatorial power of the {W}eisfeiler-{L}ehman algorithm.
\newblock In {\em Algorithms and Complexity --- 10th International Conference
  (CIAC'17) Proceedings}, volume 10236 of {\em Lecture Notes in Computer
  Science}, pages 260--271, 2017.

\bibitem{Fuerer10}
M.~Fürer.
\newblock On the power of combinatorial and spectral invariants.
\newblock {\em Linear Algebra and its Applications}, 432(9):2373 -- 2380, 2010.

\bibitem{GaoL17}
C.~Gao and J.~Lafferty.
\newblock Testing for global network structure using small subgraph statistics.
\newblock Technical report, \url{http://arxiv.org/abs/1710.00862}, 2017.

\bibitem{GrochowK07}
J.~A. Grochow and M.~Kellis.
\newblock Network motif discovery using subgraph enumeration and
  symmetry-breaking.
\newblock In {\em 11th Annual International Conference on Research in
  Computational Molecular Biology (RECOMB'07)}, volume 4453 of {\em Lecture
  Notes in Computer Science}, pages 92--106. Springer, 2007.

\bibitem{Grohe12}
M.~Grohe.
\newblock Fixed-point definability and polynomial time on graphs with excluded
  minors.
\newblock {\em J. ACM}, 59(5):27:1--27:64, 2012.

\bibitem{HasanD18}
M.~A. Hasan and V.~S. Dave.
\newblock Triangle counting in large networks: a review.
\newblock {\em Wiley Interdiscip. Rev. Data Min. Knowl. Discov.}, 8(2), 2018.

\bibitem{Hella96}
L.~Hella.
\newblock Logical hierarchies in {PTIME}.
\newblock {\em Inf. Comput.}, 129(1):1--19, 1996.

\bibitem{Higman64}
D.~Higman.
\newblock Finite permutation groups of rank 3.
\newblock {\em Math. Z.}, 86:145--156, 1964.

\bibitem{dcbook}
N.~Immerman.
\newblock {\em Descriptive complexity}.
\newblock Graduate texts in computer science. Springer, 1999.

\bibitem{KSS15}
S.~Kiefer, P.~Schweitzer, and E.~Selman.
\newblock Graphs identified by logics with counting.
\newblock In {\em Proceedings of the 40th International Symposium on
  Mathematical Foundations of Computer Science (MFCS'15)}, volume 9234 of {\em
  Lecture Notes in Computer Science}, pages 319--330. Springer, 2015.

\bibitem{KreutzerS14}
S.~Kreutzer and N.~Schweikardt.
\newblock On {H}anf-equivalence and the number of embeddings of small induced
  subgraphs.
\newblock In {\em Joint Meeting of the 23-rd {EACSL} Annual Conference on
  Computer Science Logic {(CSL)} and the 29-th Annual {ACM/IEEE} Symposium on
  Logic in Computer Science (LICS)}, pages 60:1--60:10. {ACM}, 2014.

\bibitem{hombook}
L.~Lov{\'{a}}sz.
\newblock {\em Large Networks and Graph Limits}, volume~60 of {\em Colloquium
  Publications}.
\newblock American Mathematical Society, 2012.

\bibitem{McKayP14}
B.~D. McKay and A.~Piperno.
\newblock Practical graph isomorphism, {II}.
\newblock {\em J. Symb. Comput.}, 60:94--112, 2014.

\bibitem{Milo+02}
R.~Milo, S.~Shen-Orr, S.~Itzkovitz, N.~Kashtan, D.~Chklovskii, and U.~Alon.
\newblock Network motifs: Simple building blocks of complex networks.
\newblock {\em Science}, 298(5594):824--827, 2002.

\bibitem{Morgan65}
H.~L. Morgan.
\newblock The generation of a unique machine description for chemical
  structures --- a technique developed at chemical abstracts service.
\newblock {\em J. Chem. Doc.}, 5(2):107--113, 1965.

\bibitem{MorrisKM17}
C.~Morris, K.~Kersting, and P.~Mutzel.
\newblock Glocalized {W}eisfeiler-{L}ehman graph kernels: {G}lobal-local
  feature maps of graphs.
\newblock In {\em 2017 {IEEE} International Conference on Data Mining
  (ICDM'17)}, pages 327--336. {IEEE} Computer Society, 2017.

\bibitem{Morris+18}
C.~Morris, M.~Ritzert, M.~Fey, W.~L. Hamilton, J.~E. Lenssen, G.~Rattan, and
  M.~Grohe.
\newblock Weisfeiler and {L}eman go neural: {H}igher-order graph neural
  networks.
\newblock Technical report, \url{http://arxiv.org/abs/1810.02244}, 2018.

\bibitem{Newman03}
M.~E.~J. Newman.
\newblock The structure and function of complex networks.
\newblock {\em {SIAM} Review}, 45(2):167--256, 2003.

\bibitem{Pyber14}
L.~Pyber.
\newblock Large connected strongly regular graphs are {H}amiltonian.
\newblock arXiv:1409.3041, 2014.

\bibitem{RamanaSU94}
M.~V. Ramana, E.~R. Scheinerman, and D.~Ullman.
\newblock Fractional isomorphism of graphs.
\newblock {\em Discrete Mathematics}, 132(1-3):247--265, 1994.

\bibitem{ScheinermanU97}
E.~R. Scheinerman and D.~H. Ullman.
\newblock {\em Fractional graph theory. A rational approach to the theory of
  graphs.}
\newblock Wiley: John Wiley \& Sons, 1997.

\bibitem{ShervashidzeSLMB11}
N.~Shervashidze, P.~Schweitzer, E.~J. van Leeuwen, K.~Mehlhorn, and K.~M.
  Borgwardt.
\newblock Weisfeiler-{L}ehman graph kernels.
\newblock {\em Journal of Machine Learning Research}, 12:2539--2561, 2011.

\bibitem{Toran13}
J.~Tor{\'{a}}n.
\newblock On the resolution complexity of graph non-isomorphism.
\newblock In {\em 16th International Conference on Theory and Applications of
  Satisfiability Testing ({SAT}'13)}, volume 7962 of {\em Lecture Notes in
  Computer Science}, pages 52--66. Springer, 2013.

\bibitem{UganderBK13}
J.~Ugander, L.~Backstrom, and J.~M. Kleinberg.
\newblock Subgraph frequencies: mapping the empirical and extremal geography of
  large graph collections.
\newblock In {\em 22-nd International World Wide Web Conference (WWW'13)},
  pages 1307--1318. ACM, 2013.

\bibitem{WLe68}
B.~Weisfeiler and A.~Leman.
\newblock The reduction of a graph to canonical form and the algebra which
  appears therein.
\newblock {\em NTI, Ser.~2}, 9:12--16, 1968.
\newblock English translation is available at
  \url{https://www.iti.zcu.cz/wl2018/pdf/wl_paper_translation.pdf}.

\bibitem{Wor99}
N.~Wormald.
\newblock Models of random regular graphs.
\newblock In {\em Surveys in Combinatorics}, pages 239--298. Cambridge
  University Press, 1999.

\end{thebibliography}

\appendix

\section{Proof of Lemma \ref{lem:F-free-am}}\label{app}

We prove Lemma \ref{lem:F-free-am} by splitting it into Lemmas \ref{lem:P32am}, \ref{lem:2P3am},
and \ref{lem:P322am} below. The proof is based on an explicit description
of the class of $F$-free graphs for each $F\in\{P_3+P_2,\,P_3+2P_2,\,2P_3\}$.

\subsection{Forbidden forests}

Let $\forb F$ denote the class of all graphs that do not have subgraphs isomorphic to~$F$.
Describing $\forb F$ \emph{explicitly} is a hard task in general, even if $F$ is a simple
pattern like a matching graph $sP_2$. Nevertheless, we will need explicit characterization
of $\forb F$ in several simple cases.
As the simplest fact, note that 
$$
\forb{2P_2}=\Set{K_{1,s}+tK_1,\,K_3+tK_1}_{s\ge1,\,t\ge0}.
$$
For a characterization of $\forb{3P_2}$, recall some standard graph-theoretic concepts.

The \emph{join} of graphs $G$ and $H$, denoted by $G*H$,
is obtained from the disjoint union of $G$ and $H$
by adding all possible edges between a vertex of $G$ and a vertex of~$H$.

The \emph{line graph} $L(H)$ of a graph $G$ has $E(H)$ as the set of vertices,
and $e$ and $e'$ from $E(H)$ are adjacent in $L(H)$ if and only if they share a vertex in $H$.
A \emph{clique cover} of size $k$ of a graph $G$ is a set of cliques $C_1,\ldots,C_k$ in $G$
such that $V(G)=\bigcup_{i=1}^kC_i$. Note that cliques in $L(H)$ are exactly star or
triangle subgraphs of~$H$. It follows that $H$ is in $\forb{2P_2}$ exactly when
$L(H)$ is a complete graph or, in other words, has a clique cover of size~1.
This admits an extension to the 3-matching pattern.

\begin{lemma}\label{lem:3P2}
  \begin{enumerate}[\bf 1.]
  \item 
Let $v(H)\ge6$. Then $H$ is in $\forb{3P_2}$ exactly when
$L(H)$ has a clique cover of size at most~2.
\item
Up to adding isolated vertices, 
$H\in \forb{3P_2}$ exactly in these cases:
\begin{enumerate}[\bf i.]
\item 
$v(H)\le5$,
\item
$H$ is a subgraph of one of the following graphs:
\begin{itemize}
\item 
$2K_3$,
\item 
$K_1*(K_3+sK_1)$, $s\ge0$,
\item 
$K_2*sK_1$, $s\ge1$.
\end{itemize}
\end{enumerate}
  \end{enumerate}
\end{lemma}

\noindent
Members of the last two families are depicted in Figure~\ref{fig:3P2}.
Note that $K_2*sK_1$ is the \emph{complete split} graph with the clique part of size 2.

\begin{figure}[t]
  \centering
\begin{tikzpicture}[every node/.style={circle,draw,inner sep=2pt,fill=black},scale=.8]
  \begin{scope}
    \path (0,0) node (a0) {}
      (0,1) node (a1) {} edge (a0)
   (-1,-.6) node (a2) {} edge (a0) edge (a1)
    (1,-.6) node (a3) {} edge (a0) edge (a1) edge (a2)
      (0,2) node (b1) {} edge (a1)
     (-1,2) node (b2) {} edge (a1)
      (1,2) node (b3) {} edge (a1);
\node[draw=none,fill=none,below] at (0,0) {$K_1*(K_3+3K_1)$};
  \end{scope}

  \begin{scope}[xshift=50mm,yshift=-6mm]
    \path (-1,0) node (a1) {}
           (1,0) node (a2) {} edge (a1)
          (-2,2) node (b1) {} edge (a1) edge (a2)
          (-1,2) node (b2) {} edge (a1) edge (a2)
           (0,2) node (b3) {} edge (a1) edge (a2)
           (1,2) node (b4) {} edge (a1) edge (a2)
           (2,2) node (b5) {} edge (a1) edge (a2);
\node[draw=none,fill=none,below] at (0,0) {$K_2*5K_1$};
  \end{scope}
\end{tikzpicture}

\vspace{-10mm}

  \caption{Examples of $3P_2$-free graphs.} 
  \label{fig:3P2}
\end{figure}
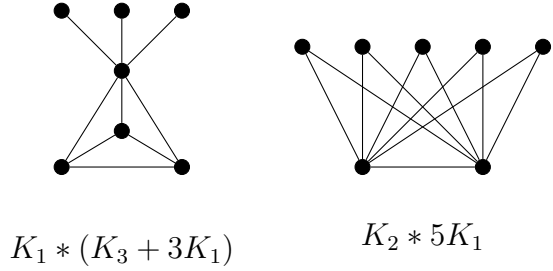

Note that Lemma \ref{lem:3P2} does not hold true without the assumption $v(H)\ge6$.
As an example, consider the complement of $P_3+2K_1$. This graph does not contain any
3-matching, while its line graph has clique cover number~3.

\begin{proof}
\textit{1.}  
In one direction, suppose that $E(H)=E_1\cup E_2$, where both $E_1$ and $E_2$ are
cliques in $L(H)$. Among any three edges $e_1,e_2,e_3$ of $H$, at least two belong to one of these cliques.
By this reason, $e_1,e_2,e_3$ cannot form a $3P_2$ subgraph of~$H$.

For the other direction, assume that $H$ does not contain any $3P_2$ subgraph.
If $H$ does not contain even any $2P_2$ subgraph, then we are done because,
as it was already mentioned, $H$ has a clique cover of size 1 in this case.
Assume, therefore, that $H$ contains two non-adjacent edges $e_1=u_1v_1$ and $e_2=u_2v_2$.
Call any vertex in $V(H)\setminus\{u_1,v_1,u_2,v_2\}$ \emph{external}.
Since $\{e_1,e_2\}$ cannot be extended to a 3-matching subgraph,
we can state the following facts.

\begin{enumerate}[(A)]
\item 
Every two external vertices in $H$ are non-adjacent.
\item 
If $u_i$ has an external neighbor $x$, then $v_i$ has no external neighbor
possibly except $x$. Symmetrically, the same holds true for $u_i$ and $v_i$ swapped.
\end{enumerate}

\noindent
We split our further analysis into three cases.

\Case1{There are external vertices $x_1$ and $x_2$ such that $x_1u_1v_1$ and $x_2u_2v_2$ are triangles.}
Claims (A) and (B) imply that $H$ has no other edge. Thus,
$H=2K_3$, and $E(H)$ is covered by two triangle cliques.

\Case2{There is a triangle $x_1u_1v_1$ and no triangle $x_2u_2v_2$.}
By Claim (B), none of the vertices $x_1$, $u_1$, and $v_1$
has an external neighbor. One of the vertices $u_2$ and $v_2$, say $u_2$ must have
at least one external neighbor $x$. If there are also other external vertices,
all of them are adjacent to $u_2$ by Claim (B).
The edges $u_1v_2$ and $v_1v_2$ are impossible in $H$ because they would form
a 3-matching together with $x_1v_1,u_2x$ and $x_1u_1,u_2x$ respectively.
Thus, $H$ can only look as shown in Figure~\ref{fig:proof-3P2}(a).
We see that $E(H)$ is covered by the triangle $x_1u_1v_1$ and the neighborhood star of~$u_2$.

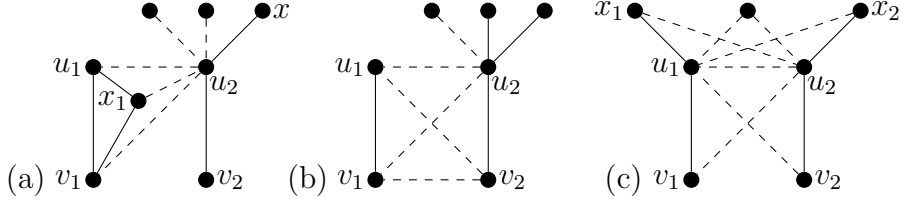
\begin{figure}[t]
  \centering
\begin{tikzpicture}[every node/.style={circle,draw,inner sep=2pt,fill=black},scale=1.5]

  \begin{scope}
    \path (0,0) node (v1) {}
      (0,1) node (u1) {} edge (v1)
      (1,0) node (v2) {}
      (1,1) node (u2) {} edge (v2) edge[dashed] (u1) edge[dashed] (v1)
    (.4,.7) node (x1) {} edge (u1) edge (v1) edge[dashed] (u2)
  (1.5,1.5) node (x)  {} edge (u2)
    (1,1.5) node (y1) {} edge[dashed] (u2)
   (.5,1.5) node (y2) {} edge[dashed] (u2);
\node[draw=none,fill=none,left] at (v1) {$v_1$};
\node[draw=none,fill=none,left] at (u1) {$u_1$};
\node[draw=none,fill=none,right] at (v2) {$v_2$};
\node[draw=none,fill=none,below right] at (u2) {$u_2$};
\node[draw=none,fill=none,left] at (x1) {$x_1$};
\node[draw=none,fill=none,right] at (x) {$x$};
\node[draw=none,fill=none] at (-.6,0) {(a)};
  \end{scope}

  \begin{scope}[xshift=25mm]
    \path (0,0) node (v1) {}
      (0,1) node (u1) {} edge (v1)
      (1,0) node (v2) {} edge[dashed] (u1) edge[dashed] (v1)
      (1,1) node (u2) {} edge (v2) edge[dashed] (u1) edge[dashed] (v1)
  (1.5,1.5) node (x)  {} edge (u2)
    (1,1.5) node (y1) {} edge (u2)
   (.5,1.5) node (y2) {} edge[dashed] (u2);
\node[draw=none,fill=none,left] at (v1) {$v_1$};
\node[draw=none,fill=none,left] at (u1) {$u_1$};
\node[draw=none,fill=none,right] at (v2) {$v_2$};
\node[draw=none,fill=none,below right] at (u2) {$u_2$};
\node[draw=none,fill=none] at (-.6,0) {(b)};
  \end{scope}

  \begin{scope}[xshift=53mm]
    \path (0,0) node (v1) {}
      (0,1) node (u1) {} edge (v1)
      (1,0) node (v2) {} edge[dashed] (u1)
      (1,1) node (u2) {} edge (v2) edge[dashed] (u1) edge[dashed] (v1)
  (1.5,1.5) node (x2) {} edge (u2) edge[dashed] (u1)
  (-.5,1.5) node (x1) {} edge (u1) edge[dashed] (u2)
   (.5,1.5) node (y)  {} edge[dashed] (u1) edge[dashed] (u2);
\node[draw=none,fill=none,left] at (v1) {$v_1$};
\node[draw=none,fill=none,left] at (u1) {$u_1$};
\node[draw=none,fill=none,right] at (v2) {$v_2$};
\node[draw=none,fill=none,below right] at (u2) {$u_2$};
\node[draw=none,fill=none,left] at (x1) {$x_1$};
\node[draw=none,fill=none,right] at (x2) {$x_2$};
\node[draw=none,fill=none] at (-.6,0) {(c)};
  \end{scope}
\end{tikzpicture}
  \caption{Proof of Lemma \ref{lem:3P2}.} 
  \label{fig:proof-3P2}
\end{figure}

\Case3{There is no external vertex $x$ such that $xu_1v_1$ or $xu_2v_2$ is a triangle.}
If one of the edges $e_1$ and $e_2$ has no external neighbor, then all edges
between $e_1$ and $e_2$ are possible and, by Claims (A) and (B),
$H$ looks as shown in Figure~\ref{fig:proof-3P2}(b).
Thus, also in this case $E(H)$ is covered by a star and a triangle (or by a star and a small
star $K_{1,t}$ with $t=1,2$).

If both $e_1$ and $e_2$ have external neighbors, then the assumption $v(H)\ge6$
implies that there is an external neighbor $x_1$ for $e_1$ and there is
an external neighbor $x_2\ne x_1$ for $e_2$. Without loss of generality,
suppose that $x_i$ is adjacent to $u_i$ for both $i=1,2$. Then $v_1$
and $v_2$ cannot be adjacent for else $H$ would contain a $3P_2$ formed by
$x_1u_1$, $x_2u_2$, and $v_1v_2$. By Claim (B), any other external
vertex is adjacent to $u_1$ or to $u_2$, or to both of them.
Thus, $H$ looks as in Figure~\ref{fig:proof-3P2}(c),
and $E(H)$ is covered by the neighborhood stars of $u_1$ and~$u_2$.

\textit{2.}  This part follows from Part 1. Let $v(H)\ge6$.
If $E(H)$ is covered by two triangles, then $H$ is a subgraph of $2K_3$.
If $E(H)$ is covered by a triangle and a star, then $H$ is a subgraph of $K_1*(K_3+sK_1)$.
Finally, if $E(H)$ is covered by two stars, then $H$ is a subgraph of $K_2*sK_1$.
\end{proof}

\noindent
As usually, $\Delta(G)$ denotes the maximum degree of a vertex in the graph~$G$.

\begin{lemma}\label{lem:forb}
Let $F\in\Set{P_3+P_2,\,P_3+2P_2,\,2P_3}$.
Up to adding isolated vertices, the classes $\forb F$ consist of the following graphs.

\begin{enumerate}[\bf 1.]
\item 
$H\in \forb{P_3+P_2}$ exactly in these cases:
\begin{enumerate}[\bf i.]
\item 
$v(H)\le4$,
\item
$\Delta(H)=1$, that is, $H=sK_2$,
\item 
$H=K_{1,s}$.
\end{enumerate}
\item 
$H\in\forb{P_3+2P_2}$ exactly in these cases:
\begin{enumerate}[\bf i.]
\item 
$v(H)\le6$,
\item 
$\Delta(H)=1$,
\item 
$H$ is a subgraph of one of the following graphs:
\begin{itemize}
\item 
$K_1*(K_3+sK_1)$, $s\ge0$,
\item 
$K_2*sK_1$, $s\ge1$.
\end{itemize}
\end{enumerate}
\item 
$H\in\forb{2P_3}$ exactly in these cases:
\begin{enumerate}[\bf i.]
\item 
$H=H_0+sK_2$, $s\ge0$, where $v(H_0)\le5$,
\item 
$H=N+sK_2$, $s\ge0$, where $N$ is the 6-vertex net graph
shown in Figure~\ref{fig:2P3},
\item 
$H$ is a subgraph of the graph $K_1*sK_2$ for some $s\ge1$.
\end{enumerate}
\end{enumerate}
\end{lemma}

\noindent
Note that Part 2.iii includes all $3P_2$-free graphs
with at least 7 vertices; see Figure~\ref{fig:3P2}.
The graphs in Part 3.ii--iii are shown in Figure~\ref{fig:2P3}.
Note also that $K_1*sK_2$ are known as \emph{friendship graphs},
and they are a part of the more general class of \emph{windmill graphs}.

\begin{figure}[t]
  \centering
\begin{tikzpicture}[every node/.style={circle,draw,inner sep=2pt,fill=black},scale=.7]
  \begin{scope}
    \path (90:1cm) node (a1) {}
         (210:1cm) node (a2) {} edge (a1)
         (330:1cm) node (a3) {} edge (a1) edge (a2)
          (90:2cm) node (b1) {} edge (a1)
         (210:2cm) node (b2) {} edge (a2)
         (330:2cm) node (b3) {} edge (a3);
\node[draw=none,fill=none] at (0,-2.5) {net graph};
  \end{scope}

  \begin{scope}[xshift=65mm]
    \path (0,0) node (a0) {}
       (75:2cm) node (a1) {} edge (a0)
      (105:2cm) node (b1) {} edge (a0) edge (a1)
      (195:2cm) node (a2) {} edge (a0)
      (225:2cm) node (b2) {} edge (a0) edge (a2)
      (315:2cm) node (a3) {} edge (a0)
      (345:2cm) node (b3) {} edge (a0) edge (a3);
\node[draw=none,fill=none] at (0,-2.5) {windmill graph $K_1*3K_2$};
  \end{scope}
\end{tikzpicture}

\vspace{-15mm}

  \caption{Examples of $2P_3$-free graphs.} 
  \label{fig:2P3}
\end{figure}
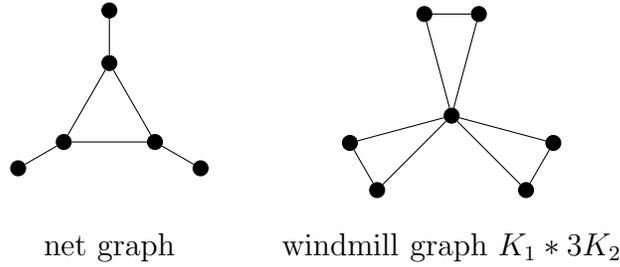

\begin{proof}
\textit{1.}
Any graph satisfying Conditions i--iii obviously does not contain $P_3+P_2$.
For the other direction, let $H$ be $(P_3+P_2)$-free.
Suppose that $\Delta(H)\ge2$. Then $H$ must be connected (recall that
we assume that $H$ has no isolated vertex). Furthermore, suppose that
$H$ has at least 5 vertices. If $\Delta(H)=2$, then $H$ is a path or a cycle
with at least 5 vertices, but all of them contain $P_3+P_2$.
It follows that $\Delta(H)\ge3$, that is, $H$ contains a 3-star $K_{1,3}$. 
Call any vertex outside this star subgraph \emph{external}.
If an external vertex is adjacent to a leaf of the 3-star, 
this clearly results in a $P_3+P_2$ subgraph. We conclude that all external vertices 
are adjacent to the central vertex $x$ of the 3-star.
Thus, $V(H)=N(x)\cup\{x\}$. Since $v(H)\ge5$, the vertex $x$ has at least 4 neighbors.
As a consequence, no two neighbors of $x$ can be adjacent, as this would yield a $P_3+P_2$.
We conclude that $H=K_{1,s}$ for some $s\ge4$.

\textit{2.}
If $H$ satisfies at least one of Conditions i--ii, it obviously does not
contain any subgraph $P_3+2P_2$. If $H$ satisfies Condition iii, then it
does not contain even $3P_2$. For the other direction, suppose
that $H$ is $(P_3+2P_2)$-free. 

If $H$ has three or more connected components, then obviously $H=sK_2$ for $s\ge3$ 
(recall the assumption that $H$ has no isolated vertex). 

Consider now the case that $H$ has exactly two connected components.
If one of them is $P_2$, then the other must (and can) be an arbitrary
connected $P_3+P_2$-free graph. By Part 1 of the lemma, the second component
has at most 4 vertices or is a star. In the former case, $v(H)\le6$.
In the latter case, $H$ is a subgraph of $K_1*(K_3+sK_1)$ for some $s$.
If none of the connected components of $H$ is $P_2$, then both of them
contain a $P_3$ and, therefore, both must be $2P_2$-free.
Since every $2P_2$-free graph is a star or a triangle, this leaves three
possibilities for $H$. If $H=2K_3$, then $v(H)\le6$.
If $H=K_3+K_{1,s}$, then $H$ is a subgraph of $K_1*(K_3+sK_1)$.
Finally, if $H=K_{1,s}+K_{1,t}$, then $H$ is a subgraph of $K_2*(s+t)K_1$.

It remains to consider the case that $H$ is connected.
Suppose that $v(H)\ge7$. Then the condition that $H$ does not
contain any $P_3+2P_2$ implies that $H$ does not contain even any
$3P_2$ subgraph. By Lemma \ref{lem:3P2}, $H$ is a subgraph of
$K_1*(K_3+sK_1)$ or $K_2*sK_1$.

\textit{3.}
Like in the previous cases, it suffices to prove the theorem in one direction.
Let $H$ be $2P_3$-free. If $H$ is disconnected, all but one connected components
must be $P_2$. As a single non-$P_2$ component, an arbitrary connected $2P_3$-free
graph is allowed. Since the class of all graphs satisfying Conditions i--iii
is closed under addition of isolated edges, it is enough to prove the claim in
the case of a connected~$H$.

Note first that $H$ contains no cycle $C_n$ for $n\ge6$ because such a cycle
contains a $2P_3$. Suppose that $v(H)\ge6$. Then $H$ contains also neither
$C_5$ nor $C_4$. Indeed, if $H$ contains a $C_5$, then $H$ must contain also
a subgraph 
\begin{tikzpicture}[every node/.style={circle,draw,inner sep=.5pt,fill=Black},scale=.25]
\path (0:1cm) node (a1) {}
     (72:1cm) node (a2) {} edge (a1) 
   (72*2:1cm) node (a3) {} edge (a2)
   (72*3:1cm) node (a4) {} edge (a3)
   (72*4:1cm) node (a5) {} edge (a4) edge (a1)
        (2,0) node (a0) {} edge (a1);
\end{tikzpicture},
which contains a $2P_3$. If $H$ contains a $C_4$, then $H$ must contain one of subgraphs
$$
\begin{tikzpicture}[every node/.style={circle,draw,inner sep=.5pt,fill=Black},scale=.4]
\path (0,0) node (a1) {}
      (1,0) node (a2) {} edge (a1) 
      (1,1) node (a3) {} edge (a2)
      (0,1) node (a4) {} edge (a3) edge (a1) 
      (2,1) node (a5) {} edge (a3)
      (2,0) node (a6) {} edge (a5);
\end{tikzpicture},
\quad
\begin{tikzpicture}[every node/.style={circle,draw,inner sep=.5pt,fill=Black},scale=.4]
\path (0,0) node (a1) {}
      (1,0) node (a2) {} edge (a1) 
      (1,1) node (a3) {} edge (a2)
      (0,1) node (a4) {} edge (a3) edge (a1) 
      (2,1) node (a5) {} edge (a3)
      (2,0) node (a6) {} edge (a3);
\end{tikzpicture},
\quad
\begin{tikzpicture}[every node/.style={circle,draw,inner sep=.5pt,fill=Black},scale=.4]
\path (0,0) node (a1) {}
      (1,0) node (a2) {} edge (a1) 
      (1,1) node (a3) {} edge (a2)
      (0,1) node (a4) {} edge (a3) edge (a1) 
      (2,1) node (a5) {} edge (a3)
      (2,0) node (a6) {} edge (a2);
\end{tikzpicture},
\quad
\begin{tikzpicture}[every node/.style={circle,draw,inner sep=.5pt,fill=Black},scale=.4]
\path (0,0) node (a1) {}
      (1,0) node (a2) {} edge (a1) 
      (1,1) node (a3) {} edge (a2)
      (0,1) node (a4) {} edge (a3) edge (a1) 
      (2,1) node (a5) {} edge (a3)
     (-1,0) node (a6) {} edge (a1);
\end{tikzpicture},
$$
and all of them contain a copy of $2P_3$. 

Suppose that $H$ contains a 3-cycle $xyz$ and call any further vertex of $H$ \emph{external}.
An external vertex can be adjacent to at most one of the vertices $x$, $y$, and $z$
for else $H$ would contain a $C_4$. Assume first that two vertices of the 3-cycle, say,
$x$ and $y$ have external neighbors $x'$ and $y'$ respectively, which must be distinct.
Since $v(H)\ge6$, there must be yet another external vertex $z'$. The only possibility
avoiding appearance of a $2P_3$ is that $z'$ is adjacent to $z$, and no other
adjacencies and further vertices are possible. Thus, in this case $v(H)=6$ and $H$ is the net graph.

Assume now that only one of the vertices of the 3-cycle, say $x$, has external neighbors.
The distance from any external vertex $v$ to $x$ is at most 2 for else a copy of $2P_3$ appears.
If this distance is equal to 2, denote the common external neighbor of $x$ and $v$ by $v'$
and note that $v'$ has degree 2 in $H$ and $v$ has degree 1.
If an external vertex $u$ does not appear in such 3-path $xv'v$, then
it is adjacent only to $x$ or, possibly, also to one vertex $u'$ of the same kind.
Then $u$ has degree 1 in $H$ in the former case and degree 2 in the latter case.
It follows that $H$ is a subgraph of 
a windmill graph $K_1*sK_2$.

It remains to consider the case that $H$ is a tree.
The diameter of $H$ is at most 4 for else $H$ would contain a $P_6$
and, hence, a $2P_3$. If the diameter is equal to 4, then $H$
contains a $P_5$ subgraph. Let $x$ be the middle vertex along this copy
of $P_5$. Note that none of the four other vertices cannot have any further neighbor
in $H$. Moreover, any branch of $H$ from $x$ can be $P_2$ or $P_3$ and nothing else.
It follows that $H$ is a subgraph of a windmill graph $K_1*sK_2$.

Suppose now that the diameter of a tree $H$ is equal to 3.
Let $x_1x_2x_3x_4$ be a copy of $P_4$ in $H$. 
Call any other vertex in $H$ \emph{external}.
Since $v(H)\ge6$, there are at least two external vertices.
None of them can be adjacent to $x_1$ or $x_4$ because then
the diameter would be larger. It is also impossible that one
external vertex is adjacent to $x_2$ and another to $x_3$
because then $H$ would contain a $2P_3$. 
Without loss of generality, assume that $x_3$ has no external
vertex. Due to the assumption on the diameter of $H$,
all the external vertices are adjacent to $x_2$.
This simple tree is obviously a subgraph of a windmill graph.

It remains to note that the trees of diameter 2 are exactly the stars
and that a single tree of diameter 1 is~$P_2$.
\end{proof}

We conclude this subsection with a straightforward characterization
of a class of graphs appearing in Lemma~\ref{lem:3P2}.2.ii and
Lemma~\ref{lem:forb}.2.iii.
The \emph{vertex cover number} $\tau(G)$ is
equal to the minimum size of a vertex cover in a graph~$G$.  

\begin{lemma}\label{lem:tau2}
  $\tau(H)\le2$ if and only if $H$ is a subgraph
of the complete split graph $K_2*sK_1$ for some $s\ge1$.
\end{lemma}

\subsection{Amenability}

In \cite{ArvindKRV17}, we defined the concept of a graph being amenable
to color refinement. Specifically, we call a graph $H$ \emph{amenable}
if \WL1 distinguishes $H$ from any other graph $G$ that is not isomorphic to $H$.
In other words, a graph is amenable if it is identifiable by \WL1 up to isomorphism.
In logical terms, a graph is amenable exactly if it is definable in the two-variable
first-order logic with counting quantifiers. 
We now show that, with just a few exceptions, every $F$-free graph for each $F$ 
from the preceding subsection is amenable.

Efficiently verifiable amenability criteria are obtained in \cite{ArvindKRV17}
and \cite{KSS15} but we do not use these powerful tools here 
as more simple and self-contained arguments are sufficient for our purposes. 
We will use the following auxiliary facts.

\begin{lemma}\label{lem:am-forests}
  \begin{enumerate}[\bf 1.]
  \item 
Every forest is amenable.
\item 
Let $K$ be a forest. Then the disjoint union $H+K$ is amenable
if and only if $H$ is amenable.
  \end{enumerate}
\end{lemma}

\noindent
Part 1 of Lemma \ref{lem:am-forests} follows from \cite[Theorem 2.5]{RamanaSU94};
see also \cite[Corollary 5.1]{ArvindKRV17}. A proof of Part 2 can be found in \cite[Section 5]{ArvindKRV17}.

A straightforward inspection shows that every graph with at most 4 vertices is amenable.
The following fact is, therefore, a straightforward consequence of Lemmas \ref{lem:forb}.1 and~\ref{lem:am-forests}.1.

\begin{lemma}\label{lem:P32am}
  Every $(P_3+P_2)$-free graph is amenable.
\end{lemma}

\noindent
Below we examine amenability of $F$-free graphs for $F\in\{3P_2,\,P_3+2P_2,\,2P_3\}$.
As the simplest application of Lemma \ref{lem:am-forests}.2,
while proving the amenability, one can always assume that the graph under
consideration has no isolated vertex.

\begin{lemma}\label{lem:vc2am}
  If $\tau(H)\le2$, then $H$ is amenable.
\end{lemma}

The subgraph of a graph $G$
induced by a subset of vertices $X\subseteq V(G)$ is denoted by
$G[X]$. For two disjoint vertex subsets $X$ and $Y$, we denote by
$G[X,Y]$ the bipartite graph with vertex classes $X$ and $Y$ and all
edges of $G$ with one vertex in $X$ and the other in $Y$.

\begin{proof}
Let $\{u,v\}$ be a vertex cover of $H$. The set $V(H)\setminus\{u,v\}$ consists of
three parts: the common neighborhood of $u$ and $v$, the neighbors solely of $u$,
and the neighbors solely of $v$; see Figure~\ref{fig:proof-vc2am}(a).
Denote the first part by $C$ and the last two parts by $A$ and $B$ respectively.

\begin{figure}[t]
  \centering
\begin{tikzpicture}[every node/.style={circle,draw,inner sep=2pt,fill=black},scale=.6]

  \begin{scope}
    \path (0,0) node (u) {}
      (2,0) node (v) {} edge[dashed] (u)
     (-1,0) node (u0) {} edge (u)
     (-1,1) node (u1) {} edge (u)
     (-1,2) node (u2) {} edge (u)
      (3,0) node (v0) {} edge (v)
      (3,1) node (v1) {} edge (v)
     (1,-1) node (w1) {} edge (u) edge (v) 
      (1,1) node (w2) {} edge (u) edge (v) 
      (1,2) node (w3) {} edge (u) edge (v);
\node[draw=none,fill=none] at ([shift={(0,-.5)}]u) {$u$};
\node[draw=none,fill=none] at ([shift={(0,-.5)}]v) {$v$};
\node[draw=none,fill=none] at (-1,3) {$A$};
\node[draw=none,fill=none] at (0,3) {$D$};
\node[draw=none,fill=none] at (1,3) {$C$};
\node[draw=none,fill=none] at (2,3) {$E$};
\node[draw=none,fill=none] at (3,3) {$B$};
\colclass{u0,u2}
\colclass{v0,v1}
\colclass{w1,w3}
\colclass{u}
\colclass{v}
\node[draw=none,fill=none] at (-2.5,.5) {(a)};
  \end{scope}

  \begin{scope}[xshift=70mm,yshift=5mm]
    \path (0,0) node (u) {}
      (2,0) node (v) {} edge[dashed] (u)
  (-.5,1.5) node (u1) {} edge (u)
    (0,1.5) node (u2) {} edge (u)
   (.5,1.5) node (u3) {} edge (u)
  (1.5,1.5) node (v1) {} edge (v)
    (2,1.5) node (v2) {} edge (v)
  (2.5,1.5) node (v3) {} edge (v)
 (-.5,-1.5) node (w1) {} edge (u) edge (v) 
   (1,-1.5) node (w2) {} edge (u) edge (v) 
 (2.5,-1.5) node (w3) {} edge (u) edge (v);
\node[draw=none,fill=none] at (4,1.5) {$A\cup B$};
\node[draw=none,fill=none] at (4,0) {$D\cup E$};
\node[draw=none,fill=none] at (4,-1.5) {$C$};
\node[draw=none,fill=none,below left] at ([shift={(-.3,0)}]u) {$u$};
\node[draw=none,fill=none,below right] at ([shift={(.3,0)}]v) {$v$};
\colclass{u,v}
\colclass{u1,v3}
\colclass{w1,w3}
\node[draw=none,fill=none] at (-2,0) {(b)};
  \end{scope}

\end{tikzpicture}
  \caption{Proof of Lemma \ref{lem:vc2am}.} 
  \label{fig:proof-vc2am}
\end{figure}

If $H$ is a forest, we are done by Lemma \ref{lem:am-forests}.1.
Suppose, therefore, that $H$ has a cycle, that is, $|C|\ge2$ or $|C|=1$
and $u$ and $v$ are adjacent. Thus, both $\deg u\ge2$ and $\deg v\ge2$,
which means that $A\cup B$ is exactly the set of all pendant\footnote{%
We call a vertex $v$ \emph{pendant} if $\deg v=1$.}
vertices in~$H$.

If $|A|\ne|B|$, then the stable partition $\Pa_H$ consists of the cells 
$A$, $B$, $C$, $D=\{u\}$, and $E=\{v\}$; see Figure~\ref{fig:proof-vc2am}(a).
In degenerate cases, $A$ or $B$
can be an empty set. A key observation is that the graph $H[X]$ for every cell $X\in\Pa_H$
is empty.
Moreover, for every two cells $X,Y\in\Pa_H$, the bipartite graph $H[X,Y]$
is either complete or empty. Suppose that $G\eqkkwl1H$. Let $A',B',C',D',E'$
be the cells of $G$ with the same stabilized colors as $A,B,C,D,E$ respectively.
We have $|X|=|X'|$ for each $X\in\Pa_H$ and its counterpart $X'\in\Pa_G$.
Moreover, Lemma \ref{lem:degree} implies that every $G[X']$ is empty and
$G[X',Y']$ is complete or empty in full accordance with $H[X,Y]$.
It follows that any bijection $f\function{V(H)}{V(G)}$ taking each cell $X$
to its counterpart $X'$ is an isomorphism from $H$ to $G$.
A similar argument will repeatedly be used throughout this subsection.

If $|A|=|B|$, then $\Pa_H$ consists of three cells 
$A\cup B$, $C$, $D\cup E$, where $C$ can be empty; see Figure~\ref{fig:proof-vc2am}(b).
Note that both $H[A\cup B]$ and $H[C]$ are empty, while
$H[D\cup E]$ is either complete or empty depending on 
adjacency of $u$ and $v$. Moreover, $H[C,D\cup E]$ is complete,
$H[C,A\cup B]$ is empty, and $H[D\cup E,A\cup B]\cong2K_{1,|A|}$,
where the two stars are centered at $u$ and $v$.
Denote the corresponding cells in a \WL1-indistinguishable graph $G$
by $(A\cup B)'$, $C'$, and $(D\cup E)'$. Lemma \ref{lem:degree} implies,
in particular, that $d((D\cup E)',(A\cup B)')=d(D\cup E,A\cup B)=|A|$
and $d((A\cup B)',(D\cup E)')=d(A\cup B,D\cup E)=1$. This determines
the graph $G[(D\cup E)',(A\cup B)']$ up to isomorphism, namely
$G[(D\cup E)',(A\cup B)']\cong2K_{1,|A|}\cong H[D\cup E,A\cup B]$.
Similarly to the preceding case, consider a cell-preserving bijection
$f\function{V(H)}{V(G)}$, assuming additionally that $f$ is
an isomorphism from $H[D\cup E,A\cup B]$ to $G[(D\cup E)',(A\cup B)']$.
As easily seen, $f$ is an isomorphism from the whole graph $H$ to~$G$.
\end{proof}

\begin{lemma}\label{lem:3P2am}
  Every $3P_2$-free graph $H$ is amenable unless $H=2C_3$.
\end{lemma}

\begin{proof}
  We use the description of the class of $3P_2$-free graph
provided by Lemma \ref{lem:3P2}.2. An easy direct inspection reveals
that all graphs with at most 5 vertices are amenable.
In view of Lemmas \ref{lem:tau2} and \ref{lem:vc2am},
it remains to consider the case that $H$ is a subgraph of 
the graph $K_1*(K_3+sK_1)$ for some $s\ge2$; see Figure~\ref{fig:3P2}.

Denote the set of pendant vertices of $K_1*(K_3+sK_1)$ by $A$.
We suppose that at least two vertices from $A$ are in $H$ for else $v(H)\le5$.
 Denote the non-pendant vertices
of $K_1*(K_3+sK_1)$ by $x_0$, $x_1$, $x_2$, and $x_3$, where $x_0$ is the common
neighbor of all pendant vertices.
By our general assumption, $H$ has no isolated vertex, which implies that $x_0$ belongs to $H$
and all vertices in $A$ are adjacent to $x_0$ in $H$.
Assume first that none of the vertices $x_1$, $x_2$, and $x_3$ is pendant in $H$.
Then $\Pa_H$ includes the cells $A$ and $\{x_0\}$.
Note that any other cell in $\Pa_H$ contains at most 3 vertices.
As easily seen, for any $X,Y\in\Pa_H$ the graph $H[X]$ is complete or empty and
 $H[X,Y]$ is a complete or empty bipartite graph.
Due to this fact, any cell-respecting bijection from $V(H)$
to the vertex set of an \WL1-indistinguishable graph $G$ provides
an isomorphism from $H$ to~$G$.

If any of $x_1$, $x_2$, and $x_3$ is pendant in $H$, then it joins the cell $A$,
and the previous argument applies.
\end{proof}

\begin{lemma}\label{lem:P322am}
  Every $(P_3+2P_2)$-free graph $H$ is amenable unless $H=2C_3$ or $H=C_6$.
\end{lemma}

\begin{proof}
  We use the description of the class $\forb{P_3+2P_2}$
provided by Lemma \ref{lem:forb}.2. A direct inspection shows
that all graphs with at most 6 vertices except $2C_3$ and $C_6$ are amenable.\footnote{%
This also follows easily from the characterization of the amenability in~\cite{ArvindKRV17}.}
Each graph $sK_2$, like any forest, is amenable. Every graph in Part iii of Lemma \ref{lem:forb}.2
is even $3P_2$-free and, therefore, amenable by Lemma~\ref{lem:3P2am}.
\end{proof}

\begin{lemma}\label{lem:2P3am}
  Every $2P_3$-free graph $H$ is amenable.
\end{lemma}

\begin{proof}
  We use the description of the class $\forb{2P_3}$
provided by Lemma \ref{lem:forb}.3.
If $H$ is as in Part i or ii, then its amenability follows from Lemma \ref{lem:am-forests}.2
and the aforementioned fact that all graphs with at most 6 vertices except $2C_3$ and $C_6$ are amenable.
Suppose, therefore, that $H$ is a subgraph of some windmill graph.

If $H$ contains no triangle, it is acyclic and amenable by Lemma \ref{lem:am-forests}.1.
If $H$ contains $C_3$ as a connected component, then $H=C_3+sK_2$, and this graph
is amenable by Lemma \ref{lem:am-forests}.2. In any other case, $H$ contains a triangle
with one vertex of degree at least 3. This is the only vertex of degree at least 3
in $H$; let us denote it by $u$. The stable partition of $H$ looks as shown
in Figure~\ref{fig:proof-2P3am}.
Specifically, $\Pa_H=\{A,B,C,D,E\}$, where $A$ consists of those neighbors of $u$
which are adjacent to another neighbor of $u$, $B$ of the pendant neighbors
of $u$, $C$ of the remaining  neighbors of $u$, $D$ of all non-neighbors of $u$
and, finally, $E=\{u\}$ ($B$, $C$, and $D$ can be empty).

\begin{figure}[t]
  \centering
\begin{tikzpicture}[every node/.style={circle,draw,inner sep=2pt,fill=black},scale=.8]
     \path (0,0) node (u) {}
      (-.25,1.5) node (b1) {} edge (u)
       (.25,1.5) node (b2) {} edge (u)
      (-1.2,1.5) node (a1) {} edge (u)
      (-1.7,1.5) node (a2) {} edge (u) edge (a1)
      (-2.3,1.5) node (a3) {} edge (u)
      (-2.8,1.5) node (a4) {} edge (u) edge (a3)
       (1.2,1.5) node (c1) {} edge (u)
       (1.7,1.5) node (c2) {} edge (u) 
       (2.2,1.5) node (c3) {} edge (u)
       (1.2,2.5) node (d1) {} edge (c1)
       (1.7,2.5) node (d2) {} edge (c2) 
       (2.2,2.5) node (d3) {} edge (c3);
\node[draw=none,fill=none,below left] at ([shift={(-.2,0)}]u) {$u$};
\node[draw=none,fill=none] at (-2,2.2) {$A$};
\node[draw=none,fill=none] at ([shift={(.7,0)}]d3) {$D$};
\node[draw=none,fill=none] at ([shift={(.7,0)}]c3) {$C$};
\node[draw=none,fill=none] at ([shift={(.7,0)}]u) {$E$};
\node[draw=none,fill=none] at (0,2.2) {$B$};

\colclass{a4,a1}
\colclass{b1,b2}
\colclass{c1,c3}
\colclass{u}
\colclass{d1,d3}

\end{tikzpicture}
  \caption{Proof of Lemma \ref{lem:2P3am}.} 
  \label{fig:proof-2P3am}
\end{figure}

Suppose that $G\eqkkwl1H$, and let $A',B',C',D',E'$ be the cells  
of $G$ with the same stabilized colors as $A,B,C,D,E$ respectively.
For any cell $X\ne A$, the graph $H[X]$ is empty.
By Lemma \ref{lem:degree}, $G[X']$ is also empty unless $X'=A'$.
Furthermore, $H[A]\cong tK_2$, where $t=|A|/2$, and hence $d(A)=1$.
By Lemma \ref{lem:degree}, $d(A')=1$ as well, which implies that
$G[A']\cong tK_2\cong H[A]$.

For any two cells $X,Y\in\Pa_H$, the bipartite graph $H[X,Y]$ is empty or
complete unless $\{X,Y\}=\{C,D\}$. By Lemma \ref{lem:degree},
the bipartite graph $G[X',Y']$ is empty or complete in the exact accordance with
$H[X,Y]$ unless $\{X',Y'\}=\{C',D'\}$. Furthermore, $H[C,D]\cong qK_2$,
where $q=|C|=|D|$, is a matching between $C$ and $D$. Therefore, $d(C,D)=d(D,C)=1$.
By Lemma \ref{lem:degree}, $d(C',D')=d(D',C')=1$, which implies that 
$G[C',D']$ is a matching between $C'$ and $D'$, and $G[C',D']\cong qK_2\cong H[C,D]$.

We construct an isomorphism $f$ from $H$ to $G$ as follows. First of all,
$f$ takes each cell $X\in\Pa_H$ onto the corresponding cell $X'\in\Pa_G$.
In particular, $f(u)=u'$, where $E'=\{u'\}$. On the cells $B$ and $C$,
the map $f$ is defined arbitrarily. On the cell $A$, the map $f$
is fixed to be an arbitrary isomorphism from $H[A]$ to $G[A']$.
Finally, $f$ is defined on $D$ so that the restriction of $f$ to $C\cup D$
is an isomorphism from $H[C,D]$ to~$G[C',D']$.
\end{proof}

\end{document}